\documentclass[thmsb,ukenglish,fleqn,11pt,letterpaper]{article}
\usepackage{amsfonts}
\usepackage{amssymb}
\usepackage{amsmath}
\usepackage{chicago}
\usepackage{epsf}
\usepackage{theorem}
\usepackage{multirow}
\usepackage{float}
\usepackage{rotating}
\usepackage{neil}
\usepackage{graphicx}
\usepackage{subcaption}
\usepackage{epstopdf}
\usepackage{color}
\usepackage[margin=2.57cm]{geometry}
\usepackage{hyperref}
\usepackage{filemod}
\usepackage{etoolbox}
\usepackage{varioref}

\restylefloat{table}
\hypersetup{backref,dvips,bookmarks=true,pdfauthor=Ju-Chen Justin Yang,colorlinks=false,breaklinks=true,hyperfigures=true,pdfstartview=FitH,linkcolor=blue,citecolor=black}
\newtheorem{theorem}{Theorem}

\newtheorem{corollary}{Corollary}

\newtheorem{lemma}{Lemma}
\newtheorem{proposition}{Proposition}
\theorembodyfont{\upshape}
\newtheorem{definition}{Definition}
\AtEndEnvironment{definition}{$\Diamond$}
\newtheorem{example}{Example}
\AtEndEnvironment{example}{$\Diamond$}

\newtheorem{remark}{Remark}
\AtEndEnvironment{remark}{$\Diamond$}

\def\bi{\begin{itemize}}
\def\ei{\end{itemize}}
\newenvironment{proof}[1][Proof]{\noindent\textbf{#1.} }{\ \rule{0.5em}{0.5em}}

\allowdisplaybreaks
\graphicspath{{../../graphs/}}
\input{tcilatex}
\begin{document}

\title{Continuous time analysis of fleeting discrete price moves}
\author{\textsc{Neil Shephard} \\
\textit{Department of Economics and Department of Statistics,}\\
\textit{Harvard University}\\
\texttt{shephard@fas.harvard.edu} \and \textsc{Justin J. Yang} \\
\textit{Department of Statistics,}\\
\textit{Harvard University}\\
\texttt{juchenjustinyang@fas.harvard.edu}}
\maketitle

\begin{abstract}
This paper proposes a novel model of financial prices where: (i) prices are
discrete; (ii) prices change in continuous time; (iii) a high proportion of
price changes are reversed in a fraction of a second. Our model is
analytically tractable and directly formulated in terms of the calendar time
and price impact curve. The resulting c\`{a}dl\`{a}g price process is a
piecewise constant semimartingale with finite activity, finite variation and
no Brownian motion component. We use moment-based estimations to fit four
high frequency futures data sets and demonstrate the descriptive power of
our proposed model. This model is able to describe the observed dynamics of
price changes over three different orders of magnitude of time intervals.
\end{abstract}



\baselineskip=20pt

\section{Introduction}

Extracting information from the order and trading flow in financial markets
is important for trading at high and low frequencies, formulating policy and
regulation and studying forensic finance. The distinctness about this area
is the frequent focus on the very short term, usually over time intervals
which may be much less than a second. At very short time scales, three
essential aspects dominate: (i) prices are discrete, due to the tick
structure of the market; (ii) prices change in continuous time; (iii) a high
proportion of price changes are fleeting, reversed in a fraction of a
second. However, the econometricians' cupboard is practically bare, for
there are nearly no models or techniques that focus on all of the three
features and put the role of the calendar time on center stage rather than
the tick time.

In this paper we develop a novel continuous calendar time framework for
prices out of a desire to capture these features in an analytically
tractable but potentially semi-parametric manner. We will show that our
model captures the serial dependence in price changes over three different
time scales: 0.1 seconds, 1 seconds, 10 seconds and 1 minute.

Although our work is a distinctive move away from the existing literature,
it will relate to a number of aspects that are often dealt with one at a
time. Here we discuss some of this material.

Most of the econometric work on the modelling of high frequency financial
data focuses on the times between trades and quote updates. This literature
splits into two: the modelling of the conditional mean duration between
events given past data and the modelling of the conditional intensity of
trade arrivals given past data. It is reviewed by, for example, \cite%
{Engle(00)}, \cite{EngleRussell(10)} and \cite{Hautsch(12)}. The former was
initiated in \cite{EngleRussell(98)} and contributions include \cite%
{ZhangRussellTsay(01)}, \cite{HamiltonJorda(02)} and \cite%
{CipolliniEngleGallo(07)}. The latter focuses around, for example, \cite%
{Russell(99)}, \cite{Bowsher(07)} and \cite{Hautsch(12)}, building on the
stochastic analysis of \cite{Hawkes(72)}.

There is much less econometric work on the discreteness of high frequency
data. Papers that focus on discreteness include \cite{RydbergShephard(03)},
\cite{RussellEngle(06)}, \cite{LiesenfeldNoltePohlmeier(06)}, \cite%
{Large(11)}, \cite{Oomen(05)}, \cite{Oomen(06)} and \cite{GriffinOomen(08)}.
Some of the early work on the impact of discreteness in practice includes
\cite{Harris(90)}, \cite{GottliebKalay(85)}, \cite{BallTorousTschoegl(85)}
and \cite{Ball(88)}. A significant approach to deal with discreteness is to
build continuous time models for prices on the positive half-line that are
then rounded to induce discreteness, sometimes with extra additive
measurement error. Examples include, for a variety of purposes, \cite%
{Hasbrouck(99bidask)}, \cite{Rosenbaum(09)}, \cite{DelattreJacod(97)}, \cite%
{Jacod(96)} and \cite{LiMykland(14)}. Also note the statistical work by \cite%
{KolassaMcCullagh(90)}.

The most comparable literatures to our own include \cite%
{BacryDelattreHoffmanMuzy(13)}, \cite{BacryDelattreHoffmanMuzy(13limit)},
\cite{FodraPham(13HFT)} and \cite{FodraPham(13)}. See also \cite%
{FauthTudor(12)}. Bacry et al. model the evolution of price changes as the
difference of two self-exciting and interacting simple counting processes.
These multivariate Hawkes processes have intensities that react to previous
moves, so an up move in the price will temporarily increase the intensity of
a down move, creating the chance that the move will turn out to be fleeting.
This elegant model only allows unit price moves, but could be extended,
while the dynamics is tightly parameterized. Fodra and Pham directly assume
an irreducible Markov chain structure on the sequence of price changes,
which is less flexible as only the current price direction will impact the
next jump direction.

Our paper has its intellectual roots in two papers. \cite%
{BarndorffNielsenPollardShephard(12)} build L\'{e}vy processes (continuous
time random walks) that are integer-valued. We are also inspired by the
stationary integer-valued processes of \cite%
{BarndorffNielsenLundeShephardVeraart(14)}. Their processes are related to
the up-stairs processes of \cite{WolpertTaqqu(05)} and the random measure
processes of \cite{WolpertBrown(11)}. Both of these processes are
stationary. \cite{BarndorffNielsenLundeShephardVeraart(14)} also bring out
the relationship between their processes and $\mathrm{M}/G/\infty $ queues
(e.g. \cite{Lindley(56)}, \cite{Reynolds(68)} and \cite[Ch. 6.31]%
{Bartlett(78)}). They also connect these models to the so-called mixed
moving average models of \cite{SurgailisRosinskiMandrekarCambanis(93)}. See
also the work of \cite{FuchsStelzer(13)}. None of these papers can be used
directly as a coherent model of high frequency data. Our paper fills this
essential gap.

Our new approach will involve events arriving in continuous time, whose
impacts on the prices may be fleeting and of variable size. The model is
directly formulated in terms of the price impact of news. Each fleeting move
is a temporary change in the price that has a random survival time until its
impact disappears. The model allows a decomposition of the discrete price
process into a continuous time random walk (due to permanent impacts) plus a
temporary fleeting component (due to market microstructure noise). The
resulting c\`{a}dl\`{a}g price process will be a piecewise constant
semimartingale with finite activity, finite variation and no Brownian motion
component. It is also capable of generating negative autocorrelations for
price changes that is consistent with the empirical observations. We have
non-parametric freedom in choosing the level of dependence in the
noises---which can even have long memory if this is needed in the data.
Alternatively, the applied researcher can tightly parameterize the model if
necessary.

In this paper our model is static: the parameters are time-invariant, not
adapting to past data. This is an important deficiency, but a stochastic
time-change can deal with most of these challenges. We will address them in
a follow up paper. Our goal here is to set down a framework that is both
empirically compelling and statistically scalable in the future work.

Finally, throughout our empirical work we have used trade prices. We could
have used our model on the best bid or ask prices. This would have had the
advantage that the best bid or ask are prices an investor can trade
immediately, while trade prices are those which have been traded by someone
in the past.

The outline of this paper is as follows. In Section \ref{sect:Levy bases} we
set up the probability structure of our model and review a couple of
building blocks from previous papers. In Section \ref{sect:price process} we
introduce the core of our contribution: defining our model for prices and
providing an analysis of this process and the corresponding return sequence.
In Section \ref{sect:moment based inference} we discuss the moment-based
estimations for these models, while in Section \ref{sect:futures data} we
apply these estimation methods to real data. Section \ref{sect:conclusion}
concludes. The Appendix has four sections. The first collects the proofs of
the various theorems given in the main text of the paper as well as the
details of some remarks. The second outlines how to compute probability mass
functions of price changes using the inverse fast Fourier transform. The
third details our data cleaning procedures. The fourth gives a
non-parametric estimator of a part of our model.

\section{Integer-valued stochastic process in continuous time\label%
{sect:Levy bases}}

\subsection{Poisson random measure}

Our framework will revolve around (i) events arriving in continuous time,
(ii) events whose impacts may be fleeting with a random survival time and
(iii) events of variable size and direction. To generate these events, it is
natural to base the underlying randomness on a three dimensional Poisson
random measure $N$ (see, e.g., \cite{Kingman(93)} for a review) with
intensity measure
\begin{equation*}
\mathbb{E}\left( N\left( \mathrm{d}y,\mathrm{d}x,\mathrm{d}s\right) \right)
=\nu (\mathrm{d}y)\mathrm{d}x\mathrm{d}s.
\end{equation*}%
Here $s$ is time (with arrivals randomly scattered on $%
\mathbb{R}
$), $x$ is random height (uniformly scattered over $\left[ 0,1\right] $),
which will be the random source for the survival time of the fleeting event,
and $y$ marks integer size (with direction) of events. These names will
become clearer later in Figure \ref{fig:Levybasis} and \ref{fig:fleeting}.
As with all Poisson random measures, the chance that there are two points
with common height or time is zero.

Price moves can be up or down, but zero is ruled out. Thus the size of the
events will be assumed to have a L\'{e}vy measure $\nu (\mathrm{d}y)$
concentrated on $y\in
\mathbb{Z}
\backslash \left\{ 0\right\} $, the non-zero integers. With no confusions,
we will sometimes abuse the notation $\nu \left( y\right) $ to denote the
mass of the L\'{e}vy measure centered at $y$, so
\begin{equation*}
\nu \left( \mathrm{d}y\right) =\sum_{y\in
\mathbb{Z}
\backslash \left\{ 0\right\} }\nu \left( y\right) \delta _{\left\{ y\right\}
}\left( \mathrm{d}y\right) ,
\end{equation*}%
where $\delta _{\left\{ y\right\} }\left( \mathrm{d}y\right) $ is the Dirac
point mass measure centered at $y$. Throughout this paper, we assume that%
\footnote{%
An equal sign with a triangle above $\triangleq $ means a definition.} $%
\left\Vert \nu \right\Vert \triangleq \int_{-\infty }^{\infty }\nu \left(
\mathrm{d}y\right) =\sum_{y\in
\mathbb{Z}
\backslash \left\{ 0\right\} }\nu \left( y\right) <\infty $.

\begin{remark}
We will see in a moment that the mass $\nu \left( y\right) $ represents the
intensity of events of size $y$, so in aggregate the L\'{e}vy measure $\nu $
will simultaneously controls the scope of all the possible jumping sizes in
addition to their individual intensity.
\end{remark}

\subsection{L\'{e}vy basis and L\'{e}vy process}

Our model will be based on the resulting homogeneous\footnote{%
Homogeneity here refers to the height and time as the points in the L\'{e}vy
basis are uniformly scatted on $\left[ 0,1\right] \times
\mathbb{R}
$.} L\'{e}vy basis on $[0,1]\times
\mathbb{R}
\longmapsto
\mathbb{Z}
\backslash \left\{ 0\right\} $, which records the size $y\in
\mathbb{Z}
\backslash \left\{ 0\right\} $ at each point in time $s\in
\mathbb{R}
$ and height $x\in \lbrack 0,1]$. It is given by%
\begin{equation*}
L(\mathrm{d}x,\mathrm{d}s)\triangleq \int_{-\infty }^{\infty }yN\left(
\mathrm{d}y,\mathrm{d}x,\mathrm{d}s\right) ,\ \ \ \ \left( x,s\right) \in %
\left[ 0,1\right] \times
\mathbb{R}
,
\end{equation*}%
and for any Borel measurable set $S\subseteq \left[ 0,1\right] \times
\mathbb{R}
$ we let%
\begin{equation*}
L\left( S\right) \triangleq \int_{\left[ 0,1\right] \times
\mathbb{R}
}1_{S}\left( x,s\right) L\left( \mathrm{d}x,\mathrm{d}s\right) ,
\end{equation*}%
where $1_{S}$ is the indicator function of $S$. To connect with the later
discussion, a L\'{e}vy process generated from $L$ can be defined as%
\begin{equation*}
L_{t}\triangleq L\left( D_{t}\right) =\int_{0}^{t}\int_{0}^{1}L(\mathrm{d}x,%
\mathrm{d}s),
\end{equation*}%
where $D_{t}\triangleq \lbrack 0,1]\times (0,t]$ is a rectangle that grows
with $t$, so $L_{t}$ just counts up the points in the L\'{e}vy basis from
time $0$ to time $t$.

\begin{example}
\label{ex: Levy basis} Suppose that $\nu (\mathrm{d}y)=\left\Vert \nu
\right\Vert \left( 0.5\times \delta _{\left\{ 1\right\} }\left( \mathrm{d}%
y\right) +0.5\times \delta _{\left\{ -1\right\} }\left( \mathrm{d}y\right)
\right) $. Then $\left\Vert \nu \right\Vert $ is the arrival rate of events
in time, each with a random height and having size $\pm 1$ with equal
probability. Figure \ref{fig:Levybasis} plots a Skellam L\'{e}vy basis $L$
using $\left\Vert \nu \right\Vert =7$, taking on $1,-1$ with black and white
dots respectively.
\begin{figure}[t]
\centering
\includegraphics{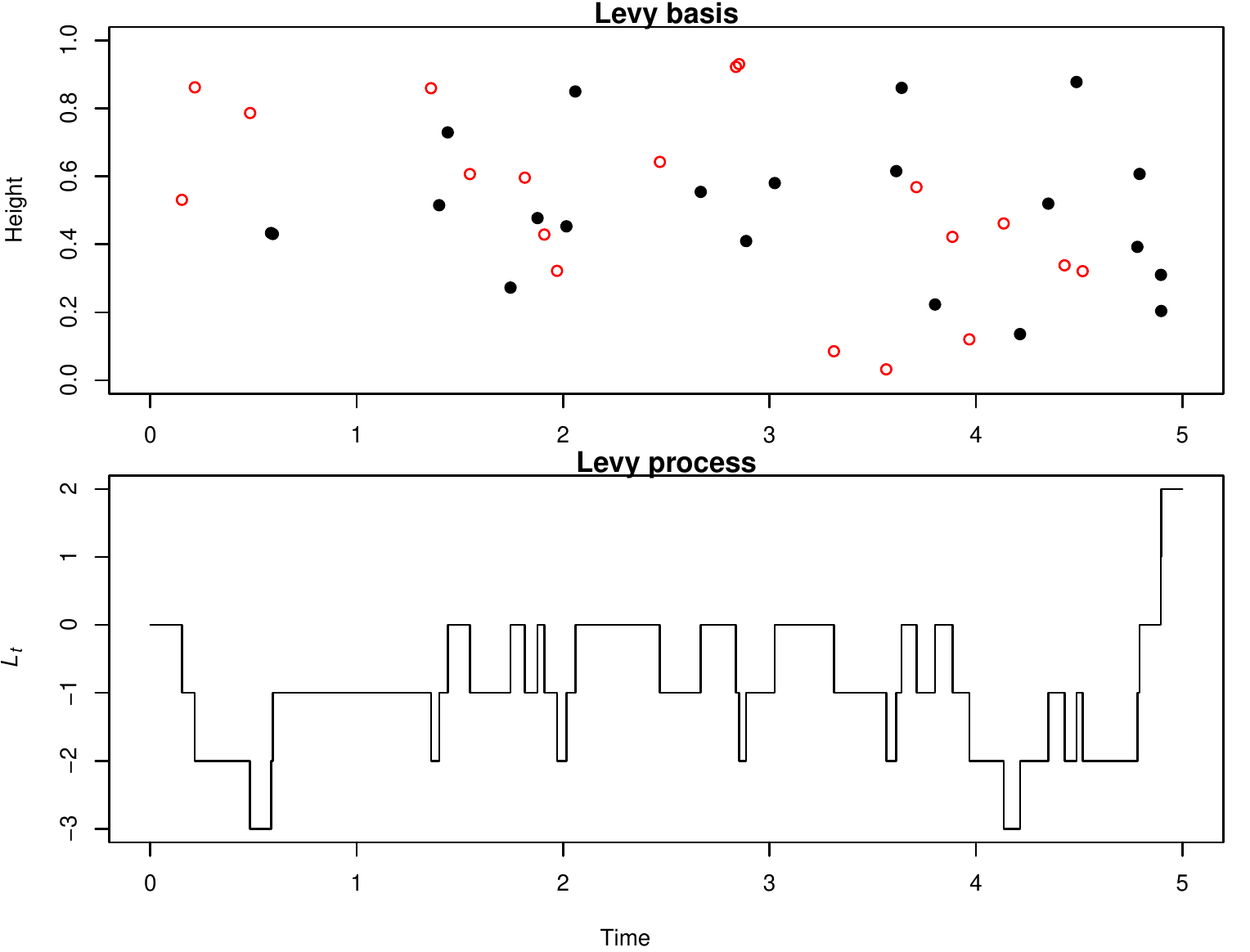}
\caption{Top: L\'{e}vy basis $L(\mathrm{d}x,\mathrm{d}s)$, where the
horizontal axis $s$ is time and the vertical axis $x$ is height, which plays
no rule in this construction of the L\'{e}vy process in the lower panel.
Black dots denote $1$, white ones $-1$. Bottom: The corresponding L\'{e}vy
process, which sums up all the effects in the L\'{e}vy basis (in the upper
panel) from time $0$ to time $t$, while the vertical axis here is the value
of the L\'{e}vy process, which jumps up by 1 by the effect of black dots and
down by 1 by white ones. Code: \texttt{LpTprocess\_Illustration.R}.}
\label{fig:Levybasis}
\end{figure}
The lower panel shows the corresponding Skellam L\'{e}vy process, which is
the difference of two independent Poisson processes with intensity $%
\left\Vert \nu \right\Vert /2$.
\end{example}

\subsection{Stationary trawl process}

To introduce fleeting moves, the random heights in the L\'{e}vy basis will
be exploited. We start from a fixed shape\footnote{%
For technical reasons, we need to assume that the fixed set $A$ is \emph{%
closed on the right} and \emph{open on the left}, that is, for every $x\in %
\left[ b,1\right] $, all the set $A\cap \left\{ \left( x,s\right) :s\leq
0\right\} $ must be a union of half-closed intervals of the form $(a,b]$.
This is enforced so the resulting jump process is \emph{c\`{a}dl\`{a}g}.
Besides, we need to assume that the projection of $A$ on the vertical axis
has Lebesgue measure $b$ so the parameter $b$ is well-defined and
statistically identifiable.} $A\subseteq \lbrack b,1]\times \left( -\infty ,0%
\right] $, where $b\in \lbrack 0,1]$ is called the permanence parameter.
Throughout we assume that the area of $A$, $leb\left( A\right) $, is finite.
\cite{BarndorffNielsenLundeShephardVeraart(14)} call $A$ a trawl for the
case of $b=0$, which is the core of their stationary integer-valued
processes. Here we call $A$ a squashed trawl, a minor variant on their idea.

\begin{definition}
\label{Def.: Squashed monotonic trawl}A squashed trawl $A$ defined by a
trawl function $d$ is obtained from%
\begin{equation*}
A\triangleq \left\{ (x,s):s\leq 0,b\leq x<d(s)\right\} ,
\end{equation*}%
where $d:\left( -\infty ,0\right] \longmapsto \lbrack b,1]$ is continuous
and monotonically increasing ($d\left( s_{1}\right) \leq d\left(
s_{2}\right) $ for all $s_{1}\leq s_{2}\leq 0$) and satisfies the following
regularity conditions: $d\left( -\infty \right) \triangleq
\lim\limits_{s\rightarrow -\infty }d\left( s\right) =b$, $d\left( 0\right)
=1 $ and $\int_{-\infty }^{0}\left( d\left( s\right) -b\right) \mathrm{d}%
s<\infty $.
\end{definition}

We now drag the set $A$ through time without changing its height%
\begin{equation*}
A_{t}\triangleq A+\left( 0,t\right) =\left\{ (x,s):s\leq t,b\leq
x<d(s-t)\right\} ,\ \ \ \ t\geq 0.
\end{equation*}%
Notice that $leb(A_{t})=leb(A)<\infty $ for all $t\geq 0$. Then the \emph{%
stationary (trawl) process} is defined as $L\left( A_{t}\right) $ for $t\geq
0$. In a moment this will be a component of our proposed price process.

\begin{example}
\label{Ex.: Stationary trawl process}The upper panel of Figure \ref%
{fig:fleeting} illustrates $A_{t}$ when $d\left( s\right) =0.5+\left(
1-0.5\right) e^{2s}$. The middle panel of Figure \ref{fig:fleeting} also
shows $L\left( A_{t}\right) $ when $L$ is a Skellam basis, which sums up all
the effects (both positive and negative) captured by (or surviving in) the
trawl. Dynamically, $L\left( A_{t}\right) $ will move up by $1$ if the
moving squashed trawl $A_{t}$ either captures one positive event that has
height above $b$ or releases a negative one; conversely, it will move down
by $1$ if vice versa. Also notice that $L\left( A_{0}\right) $ might not be
necessarily zero.
\begin{figure}[t]
\centering
\includegraphics{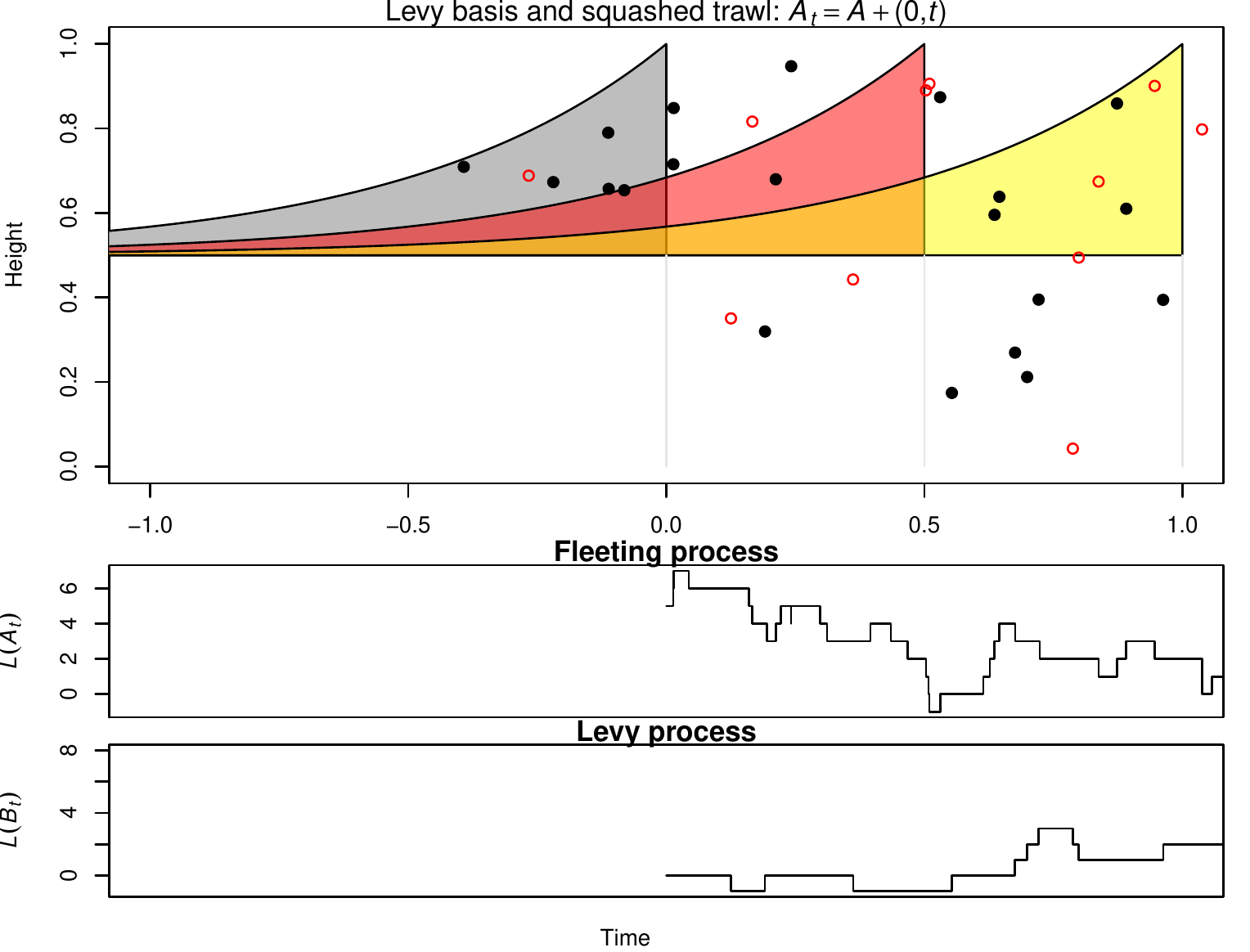}
\caption{A moving squashed trawl $A_{t}$ is joined by the L\'{e}vy basis $L(%
\mathrm{d}x,\mathrm{d}s)$, where the horizontal axis $s$ is time and the
vertical axis $x$ is height. The shaded area is an example of the trawl $A$
generated by the trawl function $d$, while we also show the outlines of $%
A_{t}$ when $t=1/2$ and $t=1$. Also shown below is the implied stationary
process $L(A_{t})$ and the L\'{e}vy process $L(B_{t})$ for $t\geq 0$, where $%
B_{t}=[0,b)\times (0,t]$. Code: \texttt{LpTprocess\_Illurstration.R}.}
\label{fig:fleeting}
\end{figure}
\end{example}

Throughout we use $\kappa _{j}(X)$ as a generic notation for the $j$-th
cumulant of an arbitrary random variable $X$. Recall that $L_{1}=L\left(
D_{1}\right) =\int_{0}^{1}\int_{0}^{1}L(\mathrm{d}x,\mathrm{d}s)$. In the
following Proposition, we rephrase the key properties of the stationary
process $L\left( A_{t}\right) $ mentioned in \cite%
{BarndorffNielsenLundeShephardVeraart(14)} under the squashed trawl variant.

\begin{proposition}
\label{prop: stat trawl}If $leb\left( A\right) <\infty $, then $L\left(
A_{t}\right) $ is well-defined and \emph{strictly stationary}. If $\kappa
_{2}\left( L_{1}\right) <\infty $ as well, then it is \emph{covariance
stationary} and for $t>s$%
\begin{equation*}
\func{Cov}\left( L\left( A_{t}\right) ,L\left( A_{s}\right) \right)
=leb\left( A_{t-s}\cap A\right) \kappa _{2}\left( L_{1}\right) ,\ \func{Cor}%
\left( L\left( A_{t}\right) ,L\left( A_{s}\right) \right) =\dfrac{leb\left(
A_{t-s}\cap A\right) }{leb\left( A\right) }.
\end{equation*}%
Furthermore, for any $t\geq 0$,
\begin{equation}
leb\left( A_{t}\cap A\right) =\int_{-\infty }^{-t}\left( d\left( s\right)
-b\right) \mathrm{d}s  \label{prop result}
\end{equation}%
is monotonically decreasing as $t$ increases.
\end{proposition}

\section{Integer-valued price process with fleeting moves\label{sect:price
process}}

\subsection{Definition}

We now turn to the main contribution of this paper. Our proposed
integer-valued price process is defined as%
\begin{equation*}
P_{t}\triangleq V_{0}+L(C_{t})=V_{0}+L\left( A_{t}\right) +L\left(
B_{t}\right) ,\ \ \ \ t\geq 0,
\end{equation*}%
where we recall that $A_{t}=A+\left( 0,t\right) $ and%
\begin{equation*}
B_{t}\triangleq \lbrack 0,b)\times (0,t],\ C_{t}\triangleq A_{t}\cup B_{t}.
\end{equation*}%
Here $V_{0}$ is a non-negative integer; $L$ is a L\'{e}vy basis; $L\left(
A_{t}\right) $ is a stationary integer-valued process that controls the
fleeting movements of the price; $V_{0}+L\left( B_{t}\right) $ is an
integer-valued L\'{e}vy process (initiating at $V_{0}$, which is aggregated
from the permanent arrivals in the past) that represents a non-stationary
component of the price process. Recall that $L\left( B_{t}\right)
=\int_{0}^{t}\int_{0}^{b}L\left( \mathrm{d}x,\mathrm{d}s\right) $, $t\geq 0$.

\begin{example}[Continued from Example \protect\ref{Ex.: Stationary trawl
process}]
\label{Ex: role of b}The lower panel of Figure \ref{fig:fleeting} shows the
corresponding Skellam L\'{e}vy process $L\left( B_{t}\right) $. Notice that
there are no permanent events in the negative time because they have been
taken into account in $V_0$. Over short time scales it is hard to tell the
difference between these two processes $L\left( A_{t}\right) $ and $L\left(
B_{t}\right) $, but over long time scales they are starkly different. For
any event arrival, if the random height $x$---not size $y$---is above $b$,
then this effect stays in $A_{t}$ temporarily and hence is fleeting; if the
height is below $b$, then this effect is always in $B_{t}$ and hence
permanent.
\end{example}

This c\`{a}dl\`{a}g price process has finite activity (i.e. finite number of
jumps in any finite interval of time, due to the L\'{e}vy basis being of
finite activity), is piecewise-constant (i.e. jumps only when there are
arrivals or departures) and consequently has finite variation. Thus the
model is in keeping with the empirical data.

\begin{remark}
\label{Cor.: Being a semimartingale}The integer-valued price process $P_{t}$
is a semimartingale with respect to its natural filtration. The details can
be found in Appendix \ref{Sec.: Proof}. Here we especially point out that a
semimartingale model that allows the fleeting behavior is atypical in the
literature of market microstructure.
\end{remark}

\begin{remark}
In this model some price moves have permanent impact. Others are fleeting,
being reversed rapidly. The lifetime of an arrival event is determined by
the trawl function. Assume that the trawl function $d$ is strictly
increasing and hence invertible. Then we can think of $G(s)\triangleq
1-d\left( -s\right) $ (with $G\left( \infty \right) \triangleq 1$) as the
cumulative distribution function of the lifetime for $s\geq 0$. Thus, for $%
U\backsim U(0,1)$, the standard uniform distribution, $G^{-1}(U)$ means the
lifetime of an arrival event with random height $U$. When $U\leq b$, then $%
G^{-1}\left( U\right) =\infty $, meaning it is permanent. For $U>b$ then the
event will last $G^{-1}(U)<\infty $, meaning it is fleeting.
\end{remark}

\begin{remark}
If a new piece of news arrives at time $t$, it impacts the price through the
arrival of a new point in the L\'{e}vy basis. For concreteness of exposition
here, suppose it has unit impact. Then the expected impact of this
individual event at time $t+s$ is $d(-s)$, where $s\geq 0$. Hence the trawl
function directly describes the \emph{price impact curve} of news arrivals.
It is tempting to label $d$ the price impact function, but we continue with
the trawl nomenclature. The permanent impact of the unit news is thus $b$.
\end{remark}

\subsection{Distribution of price changes}

The following Theorem characterizes the distribution of price changes over a
time length $t$.

\begin{theorem}
\label{thm:dist of returns}Let $A\backslash B$ be set subtraction (all
elements of $A$ except those that are also in $B$). Then
\begin{equation*}
P_{t}-P_{0}=L\left( C_{t}\right) -L\left( C_{0}\right) =L\left(
C_{t}\backslash C_{0}\right) -L(C_{0}\backslash C_{t}),
\end{equation*}%
where $L\left( C_{t}\backslash C_{0}\right) $ is independent of $%
L(C_{0}\backslash C_{t})$. Consequently the logarithmic characteristic
function of returns is
\begin{eqnarray*}
C\left( \theta \ddagger P_{t}-P_{0}\right) &=&btC\left( \theta \ddagger
L_{1}\right) +leb(A_{t}\backslash A)\left( C\left( \theta \ddagger
L_{1}\right) +C\left( -\theta \ddagger L_{1}\right) \right) ,\quad \text{%
where} \\
C\left( \theta \ddagger L_{1}\right) &\triangleq &\log \mathbb{E}\left( e^{%
\mathrm{i}\theta L_{1}}\right) ,\ \ \ \ \mathrm{i}\triangleq \sqrt{-1},\ \ \
\ L_{1}=\int_{0}^{1}\int_{0}^{1}L(\mathrm{d}x,\mathrm{d}s).
\end{eqnarray*}%
Furthermore, if the $j$-th cumulant of $L_{1}$ exists, then%
\begin{eqnarray*}
\kappa _{j}(P_{t}-P_{0}) &=&bt\kappa _{j}(L_{1}),\ \ \ \ j=1,3,5,..., \\
\kappa _{j}(P_{t}-P_{0}) &=&\left( bt+2leb\left( A_{t}\backslash A\right)
\right) \kappa _{j}(L_{1}),\ \ \ \ j=2,4,6,....
\end{eqnarray*}
\end{theorem}

\begin{remark}
Notice that $C_{t}\backslash C_{0}$ has the physical interpretation of
arrivals during the time period $0$ to $t$ for both positive and negative
effects; $C_{0}\backslash C_{t}$ are departures instead. Further, the
equalities
\begin{equation}
leb(A_{t}\backslash A)=leb(A)-leb(A_{t}\cap A)=leb(A\backslash
A_{t})=\int_{-t}^{0}\left( d\left( s\right) -b\right) \mathrm{d}s
\label{trawl increament area}
\end{equation}%
are often helpful in calculations.
\end{remark}

\begin{remark}
The probability mass function of $P_{t}-P_{0}$ can be computed using the
characteristic function and the inverse fast Fourier transform. The details
can be found in the Appendix \ref{Sec.: IFFT Details}.
\end{remark}

\begin{remark}
Even though our model is written down for the study of high frequency data,
it can easily connect back to those diffusion based models that are commonly
used to study data at less high frequency. Theorem \ref{thm:dist of returns}
further implies that the fleeting price process becomes a Brownian motion at
lower frequency. Precisely, if $\kappa _{2}\left( L_{1}\right) <\infty $ and
$X_{t}^{\left( c\right) }\triangleq c^{-1/2}\left( P_{ct}-P_{0}-bct\kappa
_{1}\left( L_{1}\right) \right) $, then $X_{\cdot }^{\left( c\right) }%
\overset{\mathcal{L}}{\rightarrow }W_{\cdot }$ as $c\rightarrow \infty $,
where $W_{\cdot }$ is a Wiener process or a standard Brownian motion.
\end{remark}

Let $\Delta P_{t}\triangleq P_{t}-P_{t-}$ be the instantaneous jump (or
return) of the price process at time $t$. By the instantaneous jumping
distribution, we mean the probability of $\Delta P_{t}=y$ given that $\Delta
P_{t}\neq 0$ for $y\in
\mathbb{Z}
\backslash \left\{ 0\right\} $. In the following we give a closed-form
expression for this distribution.

\begin{theorem}
\label{Thm.: Jumping distribution}The instantaneous jumping distribution is%
\begin{equation}
\mathbb{P}\left( \Delta P_{t}=y|\Delta P_{t}\neq 0\right) =\dfrac{\nu \left(
y\right) +\nu \left( -y\right) \left( 1-b\right) }{\left( 2-b\right)
\left\Vert \nu \right\Vert }.  \label{Jumping distribution}
\end{equation}
\end{theorem}

Notice that the trawl function $d$ in the fleeting component has no impact
on the instantaneous jumping distribution: what is important is $b$, which
controls the amount of potential departures among all the arrival jumps.
Besides, the left-hand side of equation (\ref{Jumping distribution}) can be
easily estimated from the data, so we might in turn estimate $\nu $ and $b$
by simple moment-matching. To calibrate the trawl, Theorem \ref{thm:dist of
returns} imply the easy use of sample cumulants across different $t$ to
infer the shape of $leb\left( A_{t}\backslash A\right) $ and hence $d$. We
will see these in Section \ref{sect:moment based inference} later.

\subsection{Autocorrelation structure of price changes}

Theorem \ref{Thm.: dependence of returns} captures the linear dependence in
the price changes.

\begin{theorem}
\label{Thm.: dependence of returns}Assume that $\kappa _{2}(L_{1})<\infty $.
Then the price changes have the autocorrelation structure, for some sampling
interval $\delta >0$ and $k=1,2,...$%
\begin{eqnarray*}
\gamma _{k} &\triangleq &\mathrm{Cov}\left( \left( P_{(k+1)\delta
}-P_{k\delta }\right) ,(P_{\delta }-P_{0})\right) \\
&=&\left( leb(A_{(k+1)\delta }\backslash A)-2leb(A_{k\delta }\backslash
A)+leb(A_{(k-1)\delta }\backslash A)\right) \kappa _{2}(L_{1}), \\
\rho _{k} &\triangleq &\mathrm{Cor}\left( \left( P_{(k+1)\delta }-P_{k\delta
}\right) ,(P_{\delta }-P_{0})\right) \\
&=&\frac{leb(A_{(k+1)\delta }\backslash A)-2leb(A_{k\delta }\backslash
A)+leb(A_{(k-1)\delta }\backslash A)}{b\delta +2leb(A_{\delta }\backslash A)}%
.
\end{eqnarray*}
\end{theorem}

\begin{corollary}
\label{Cor.: Negative Correlation}$\rho _{k}\leq 0$ for all $k=1,2,...$.
This inequality becomes strict when $d$ is strictly increasing (i.e. $%
d\left( s_{1}\right) <d\left( s_{2}\right) $ for all $s_{1}<s_{2}\leq 0$).
\end{corollary}

\begin{remark}
For a pure L\'{e}vy process ($b=1$), $leb\left( A_{t}\backslash A\right) =0$
for all $t$, so clearly $\rho _{k}=0$ for all $k=1,2,...$, as expected. On
the other hand, equation (\ref{trawl increament area}) implies%
\begin{equation*}
\lim\limits_{\delta \rightarrow 0}\dfrac{leb\left( A_{l\delta }\backslash
A\right) }{\delta }=\left( 1-b\right) l,\ \lim\limits_{\delta \rightarrow
\infty }leb\left( A_{l\delta }\backslash A\right) =leb\left( A\right) ,\ \ \
\ l=1,2,...,
\end{equation*}%
so it is easy to see that, for any fixed $k=1,2,...$,%
\begin{equation*}
\lim\limits_{\delta \rightarrow 0}\rho _{k}=\lim\limits_{\delta \rightarrow
\infty }\rho _{k}=0.
\end{equation*}%
Thus, Corollary \ref{Cor.: Negative Correlation} implies that $\rho _{k}$ is
not a monotonic function of the sampling interval $\delta $. This matches
with the empirical data, which we will see later in Figure \ref%
{fig.:Empirical Autocorrelation}.
\end{remark}

\subsection{Power variation}

Quadratic variation plays a central role in stochastic analysis and modern
finance (e.g. \cite{AndersenBollerslevDieboldLabys(01)} and \cite%
{BarndorffNielsenShephard(02realised)}). For any $r\geq 0$, we define the $r$%
-th power L\'{e}vy basis as%
\begin{equation*}
\Sigma (\mathrm{d}x,\mathrm{d}s;r)\triangleq \int_{-\infty }^{\infty
}\left\vert y\right\vert ^{r}N(\mathrm{d}y,\mathrm{d}x,\mathrm{d}s)
\end{equation*}%
with mean measure%
\begin{equation*}
\mu (\mathrm{d}x,\mathrm{d}s;r)\triangleq \mathrm{d}x\mathrm{d}%
s\int_{-\infty }^{\infty }\left\vert y\right\vert ^{r}\nu (\mathrm{d}y),
\end{equation*}%
assuming that $\int_{-\infty }^{\infty }\left\vert y\right\vert ^{r}\nu (%
\mathrm{d}y)<\infty $. Theorem \ref{Thm:QV} relates $\Sigma $ to
\begin{equation*}
\{P\}_{t}^{[r]}=\lim\limits_{\delta \rightarrow 0}\sum_{k=1}^{t/\delta
}\left\vert P_{k\delta }-P_{\left( k-1\right) \delta }\right\vert
^{r}=\sum_{0<s\leq t}\left\vert \Delta P_{s}\right\vert ^{r},
\end{equation*}%
the $r$-th (unnormalized) power variation, which was formalized in finance
by \cite{BarndorffNielsenShephard(04jfe)}. The special case of $r=2$ yields
the quadratic variation. Notice that in our model we can compute $%
\{P\}_{t}^{[r]}$ exactly, just using the price path. It is finite for all $%
r\geq 0$ with probability one. This contrasts with the vast majority of work
in econometrics that would take $\{P\}_{t}^{[r]}$ as infinity due to the
impact of market microstructure.

\begin{theorem}
\label{Thm: Expectation of Power variation}\label{Thm:QV}For any $r\geq 0$,
the $r$-th power variation is
\begin{eqnarray*}
\left\{ P\right\} _{t}^{[r]} &=&\Sigma (B_{t};r)+Z_{t}^{[r]},\ \ \ \
B_{t}\triangleq \lbrack 0,b]\times (0,t], \\
Z_{t}^{[r]} &\triangleq &\Sigma (H_{t};r)+\Sigma (G_{t};r),\ \ \ \
H_{t}\triangleq \lbrack b,1]\times (0,t],\ G_{t}\triangleq \left( H_{t}\cup
A\right) \backslash A_{t}.
\end{eqnarray*}%
Furthermore, their expectations are%
\begin{equation}
\mathbb{E}\left( \left\{ P\right\} _{t}^{[r]}\right) =\left( 2-b\right)
t\int_{-\infty }^{\infty }\left\vert y\right\vert ^{r}\nu (\mathrm{d}y).
\label{Expectation of the power variation}
\end{equation}
\end{theorem}

\begin{remark}
Like (\ref{Jumping distribution}), (\ref{Expectation of the power variation}%
) does not feature the trawl function, as each arrival is joined by a
departure. Hence it is always robust to the details of $d$. Further,%
\begin{equation*}
\mathbb{E}\left( \{P\}_{t}^{[r]}\right) =\mathbb{E}\left(
\{P\}_{t}^{[0]}\right) \int_{-\infty }^{\infty }\left\vert y\right\vert ^{r}%
\dfrac{\nu \left( \mathrm{d}y\right) }{\left\Vert \nu \right\Vert }.
\end{equation*}%
Notice that $\{P\}_{t}^{[0]}$ counts the total number of jumps of the
process $P$ up to time $t$, so throughout we call it the counting process of
price moves. It will also play an important role in Section \ref{sect:moment
based inference} for the construction of our moment-based estimate for the
model parameters.
\end{remark}

We think of the random $Z_{t}^{[r]}$, which is finite with probability one,
as the component of power variation due to fleeting moves in prices, for
\begin{equation*}
\left\{ P\right\} _{t}^{[r]}-\left\{ L\left( B_{t}\right) \right\}
_{t}^{[r]}=Z_{t}^{[r]}
\end{equation*}%
is the asymptotic stochastic bias of the power variation.

High frequency econometricians would typically think of terms like $%
Z_{t}^{[2]}$ as the driver of the bias in realized variance due to market
microstructure effects (e.g. \cite{HansenLunde(06)}, \cite{Zhang(06)}, \cite%
{JacodLiMyklandPodolskijVetter(07)}, \cite{MyklandZhang(12semstat)} and \cite%
{BarndorffNielsenHansenLundeShephard(08realised)}), but it is typically
infinite in their studies while here and empirically it is finite with
probability one\footnote{%
Econometricians use a variety of models for market microstructure noise.
Typically the noise appears each time a trade happens, e.g. in \cite%
{Zhou(96)} the noises are i.i.d. with a zero mean. Hence we can think of
these types of models as purely statistical measurement error models. In
more recent times, the i.i.d. assumption has been generalized to allow some
levels of temporal dependence and volatility clustering, but all in tick
time instead of the calendar time. All of these models of noise have the
power variations being infinity. There is another set of papers that think
of prices as being a rounded version of a semimartingale. This is closer to
our paper, but here the level of dependence in price moves is entirely
dependent on the size of the ticks in comparison to the volatility of the
semimartingale. This is insufficiently flexible to fit the data. Another set
of papers round a semimartingale with additive measurement noise, but again
this has infinite power variation, which does not coincide with the
empirical observations.}.

We recall from Theorem \ref{thm:dist of returns} that%
\begin{equation*}
\mathbb{E}(P_{t}-P_{0})=bt\kappa _{1}(L_{1}),\ \ \ \ \func{Var}%
(P_{t}-P_{0})=\left\{ bt+2leb\left( A_{t}\backslash A\right) \right\} \kappa
_{2}\left( L_{1}\right) .
\end{equation*}%
We now think about returns over the time interval $\left[ 0,T\right] $, so
the realized variance is%
\begin{equation*}
RV^{\left( n\right) }\triangleq \sum_{k=1}^{n}(P_{k\delta _{n}}-P_{\left(
k-1\right) \delta _{n}})^{2},\ \ \ \ \delta _{n}\triangleq \dfrac{T}{n}.
\end{equation*}

\begin{proposition}
\label{Thm: RV expectation}Assume that $\kappa _{2}(L_{1})<\infty $. Then%
\begin{equation*}
\mathbb{E}\left( RV^{\left( n\right) }\right) =\left( b+2\dfrac{leb\left(
A_{\delta _{n}}\backslash A\right) }{\delta _{n}}\right) T\kappa _{2}\left(
L_{1}\right) +b^{2}T\delta _{n}\kappa _{1}^{2}\left( L_{1}\right).
\end{equation*}
\end{proposition}

We can set the context of Proposition \ref{Thm: RV expectation} by
discussing the two extremes $n=1$ and $n\rightarrow \infty $ for a large $T$%
. For $n=1$, as $T\rightarrow \infty $,%
\begin{eqnarray*}
\mathbb{\mathbb{E}}\left( RV^{\left( 1\right) }\right)  &=&\left( b+2\dfrac{%
leb\left( A_{T}\backslash A\right) }{T}\right) T\kappa _{2}\left(
L_{1}\right) +b^{2}T^{2}\kappa _{1}^{2}\left( L_{1}\right)  \\
&\approx &bT\kappa _{2}\left( L_{1}\right) +b^{2}T^{2}\kappa _{1}^{2}\left(
L_{1}\right)  \\
&=&\kappa _{2}\left( L\left( B_{T}\right) \right) +\left( \kappa _{1}\left(
L\left( B_{T}\right) \right) \right) ^{2}=\mathbb{E}\left( L\left(
B_{T}\right) ^{2}\right) ,
\end{eqnarray*}%
where the second line uses $leb\left( A_{T}\backslash A\right) \approx
leb\left( A\right) $. For $n\rightarrow \infty $ and a fixed $T$,%
\begin{equation*}
\lim\limits_{n\rightarrow \infty }\mathbb{E}\left( RV^{\left( n\right)
}\right) =\left( 2-b\right) T\kappa _{2}\left( L_{1}\right) .
\end{equation*}%
Therefore, in this model the realized variance and the volatility of price
changes are highly distorted by the fleeting component. A variance signature
plot ($RV^{\left( T/\delta \right) }$ against $\delta $) for our model will
start out high around $\left( 2-b\right) T\kappa _{2}\left( L_{1}\right) $
(the expected quadratic variation of the price process) for large $n$ (dense
sampling) and tend downwards to approximately $bT\kappa _{2}\left(
L_{1}\right) $ (the expected quadratic variation of the L\'{e}vy process
component, assuming that $\kappa _{1}\left( L_{1}\right) $ being very
small). A minor variant of this type of plots, which we will discuss in
Remark \ref{Rmk.: Variant of the variance signature plot}, can be found in
Figure \ref{fig.:Empirical Variance Signature Plot} later in our empirical
work.

\subsection{Generalized compound representation\label{Sec.: Generalized
compound representation}}

As the price process is of finite activity, it can be usefully written as a
generalized compound process, driven by the counting process of price moves.
Here we detail this. First recall that $G(s)=1-d\left( -s\right) $ (with $%
G\left( \infty \right) =1$) denotes the cumulative distribution function of
the lifetime for $s\geq 0$.

$L(A_{0})$ is built out of $N^{A\ast }$ initial surviving events, who arrive
at times $\tau _{1}^{A\ast }<...<\tau _{N^{\ast }}^{A\ast }\leq 0$ and jump
with sizes $\varkappa _{1}^{A\ast },...,\varkappa _{N^{\ast }}^{A\ast }$.
Each arrival has a lifetime $G^{-1}(U_{1}^{A\ast }),...,G^{-1}(U_{N^{A\ast
}}^{A\ast })$, where $\tau _{j}^{A\ast }+G^{-1}(U_{j}^{A\ast })>0$ and $%
U_{j}^{A\ast }\overset{\text{i.i.d.}}{\sim }U(b,1)$. Thus we can write $%
L(A_{0})=\sum_{j=1}^{N^{A\ast }}\varkappa _{j}^{A\ast }1_{\tau _{j}^{A\ast
}+G^{-1}(U_{j}^{A\ast })>0}$. When $\varkappa _{j}^{A\ast }=1$ for all $j$,
this representation has a close connection to a $\mathrm{M}/G/\infty $ queue
(i.e. Markov arrivals, with a fixed service time distribution $G$, but with
an infinite number of servers).

As time progresses some events die and the initial values thin down to $%
\sum_{j=1}^{N^{A\ast }}\varkappa _{j}^{A\ast }1_{\tau _{j}^{A\ast
}+G^{-1}(U_{j}^{A\ast })>t}$ while new ones are born $\sum_{j=1}^{N_{t}^{A}}%
\varkappa _{j}^{A}1_{\tau _{j}^{A}+G^{-1}(U_{j}^{A})>t}$, where $N_{t}^{A}$
is the number of births from time $0$ to time $t$ with heights greater than $%
b$. The corresponding $\tau _{j}^{A}$'s and $\varkappa _{j}^{A}$'s are the
arrival times of these events and size of the moves. Thus the stationary
process is%
\begin{equation*}
L\left( A_{t}\right) =\sum_{j=1}^{N^{A\ast }}\varkappa _{j}^{A\ast }1_{\tau
_{j}^{A\ast }+G^{-1}(U_{j}^{A\ast })>t}+\sum_{j=1}^{N_{t}^{A}}\varkappa
_{j}^{A}1_{\tau _{j}^{A}+G^{-1}(U_{j}^{A})>t},\ \ \ \ t\geq 0.
\end{equation*}%
The corresponding impact of the permanent changes is a compound Poisson
process $L\left( B_{t}\right) =\sum_{j=1}^{N_{t}^{B}}\varkappa _{j}^{B}$,
where $N_{t}^{B}$ counts the number of permanent arrivals up to time $t$ and
$\tau _{j}^{B}$'s and $\varkappa _{j}^{B}$'s are the corresponding arrival
times and jump sizes. We also write $\tau _{k}$ to be any one of the jumping
times from resulted chronologically from both the arrivals and departures;
similarly for $\varkappa _{k}$. Then $N_{t}\triangleq \#\left\{ k:\tau
_{k}\leq t\right\} $ counts the total number of jumps of the price process
up to time $t$.

All these imply that%
\begin{equation}
P_{t}=V_{0}+\sum_{j=1}^{N^{A\ast }}\varkappa _{j}^{A\ast }1_{\tau
_{j}^{A\ast }+G^{-1}(U_{j}^{A\ast })>t}+\sum_{j=1}^{N_{t}^{A}}\varkappa
_{j}^{A}1_{\tau
_{j}^{A}+G^{-1}(U_{j}^{A})>t}+\sum_{j=1}^{N_{t}^{B}}\varkappa
_{j}^{B}=P_{0}+\sum_{k=1}^{N_{t}}\varkappa _{k}.  \label{simple presentation}
\end{equation}%
Equation (\ref{simple presentation}) is called a generalized compound
representation. It links with the very large literature on the use of
compound Poisson processes in financial econometrics, e.g. \cite{Press(67)}.
However, here we allow a fraction of the jumps to be fleeting, so the
resulting counting process $N_{t}$ is not simply a Poisson process.

\subsection{Parameterized trawl function\label{sect:parameterised trawl}}

To fit this type of model using data, it is sometimes helpful to index the
trawl function by a small number of parameters. Throughout we work within
the following framework.

\begin{definition}
A superposition trawl function has%
\begin{equation}
d(s)=b+(1-b)\int_{0}^{\infty }e^{\lambda s}\pi \left( \mathrm{d}\lambda
\right) ,\ \ \ \ s\leq 0,  \label{superposition}
\end{equation}%
where $\pi $ is an arbitrary probability measure on $(0,\infty )$. We
constrain the superposition class to where $\int_{0}^{\infty }\lambda
^{-1}\pi \left( \mathrm{d}\lambda \right) <\infty $.
\end{definition}

Whatever the probability measure $\pi $ the resulting $d$ always exists
since $0\leq \int_{0}^{\infty }e^{\lambda s}\pi \left( \mathrm{d}\lambda
\right) \leq \int_{0}^{\infty }\pi \left( \mathrm{d}\lambda \right) =1$, as $%
s\leq 0$. The constraint $\int_{0}^{\infty }\lambda ^{-1}\pi \left( \mathrm{d%
}\lambda \right) <\infty $ is needed to ensure that the area of $A$ is
finite, for this area is%
\begin{eqnarray}
leb(A) &=&\int_{-\infty }^{0}\int_{b}^{d(s)}\mathrm{d}x\mathrm{d}%
s=\int_{-\infty }^{0}\left( d(s)-b\right) \mathrm{d}s=(1-b)\int_{-\infty
}^{0}\int_{0}^{\infty }e^{\lambda s}\pi \left( \mathrm{d}\lambda \right)
\mathrm{d}s  \notag \\
&=&(1-b)\int_{0}^{\infty }\int_{-\infty }^{0}e^{\lambda s}\mathrm{d}s\pi
\left( \mathrm{d}\lambda \right) =\left( 1-b\right) \int_{0}^{\infty }\frac{1%
}{\lambda }\pi \left( \mathrm{d}\lambda \right) .  \label{leb of A}
\end{eqnarray}

Using equation (\ref{prop result}) the superposition framework (\ref%
{superposition}) has%
\begin{equation*}
leb(A_{t}\cap A)=\left( 1-b\right) \int_{0}^{\infty }\frac{e^{-t\lambda }}{%
\lambda }\pi \left( \mathrm{d}\lambda \right) ,\ \ \ \ t\geq 0,
\end{equation*}%
so, combining it with equation (\ref{leb of A}), we have%
\begin{equation*}
\int_{0}^{\infty }\mathrm{Cor}(L\left( A_{t}\right) ,L(A_{0}))\mathrm{d}t=%
\dfrac{\int_{0}^{\infty }\lambda ^{-2}\pi \left( \mathrm{d}\lambda \right) }{%
\int_{0}^{\infty }\lambda ^{-1}\pi \left( \mathrm{d}\lambda \right) }.
\end{equation*}%
Thus, the superposition trawl has long memory if and only if $%
\int_{0}^{\infty }\lambda ^{-2}\pi \left( \mathrm{d}\lambda \right) =\infty $%
.

In the following we focus only on choices of specific $\pi $. These special
cases have been analyzed in \cite{BarndorffNielsenLundeShephardVeraart(14)},
so here we only state them to establish notation for our applied work.

\begin{example}
\label{ex: Exp trawl}When $\pi $ has a single atom of support at $\lambda >0$%
, this is the exponential trawl%
\begin{eqnarray}
d(s) &=&b+(1-b)\exp (\lambda s),\ \ \ \ s\leq 0,  \label{OU trawl} \\
leb\left( A\right)  &=&\dfrac{1-b}{\lambda },\ leb\left( A_{t}\cap A\right) =%
\dfrac{1-b}{\lambda }e^{-\lambda t}.  \notag
\end{eqnarray}%
Trivially it only allows short memory as $\int_{0}^{\infty }\bar{\lambda}%
^{-2}\pi \left( \mathrm{d}\bar{\lambda}\right) =\lambda ^{-2}<\infty $
whenever $\lambda >0$.
\end{example}

\begin{example}
When%
\begin{equation*}
\pi \left( \mathrm{d}\lambda \right) =b+(1-b)\dfrac{\alpha ^{H}}{\Gamma
\left( H\right) }\lambda ^{H-1}e^{-\lambda \alpha }\mathrm{d}\lambda ,\ \ \
\ \alpha >0,\ H>1,
\end{equation*}%
we produce the superposition gamma (sup-$\Gamma $) trawl%
\begin{eqnarray}
d(s) &=&b+(1-b)\left( 1-\frac{s}{\alpha }\right) ^{-H},\ \ \ \ s\leq 0,
\label{long memory trawl} \\
leb\left( A\right) &=&\left( 1-b\right) \dfrac{\alpha }{H-1},\ leb(A_{t}\cap
A)=\frac{(1-b)\alpha }{H-1}\left( 1+\frac{t}{\alpha }\right) ^{1-H},\ \ \ \
t\geq 0.  \notag
\end{eqnarray}%
It has long memory when $H\in (1,2]$ and short memory when $H>2$ as
\begin{equation*}
\int_{0}^{\infty }\lambda ^{-2}\pi \left( \mathrm{d}\lambda \right) =\dfrac{%
\Gamma \left( H-2\right) }{\Gamma \left( H\right) }<\infty \text{ if and
only if }H>2.
\end{equation*}
\end{example}

\begin{example}
When%
\begin{equation*}
\pi \left( \mathrm{d}\lambda \right) =b+(1-b)\dfrac{\left( \gamma /\delta
\right) ^{\nu }}{2K_{\nu }\left( \gamma \delta \right) }\lambda ^{\nu
-1}e^{-\left( \gamma ^{2}\lambda +\delta ^{2}\lambda ^{-1}\right) /2}\mathrm{%
d}\lambda ,\ \ \ \ \gamma ,\delta >0,\ \nu \in
\mathbb{R}
,
\end{equation*}%
we produce the superposition generalized inverse Gaussian (sup-GIG) trawl%
\begin{eqnarray}
d\left( s\right) &=&b+(1-b)\left( 1-\dfrac{2s}{\gamma ^{2}}\right) ^{-\nu /2}%
\dfrac{K_{\nu }\left( \gamma \delta \sqrt{1-2s/\gamma ^{2}}\right) }{K_{\nu
}\left( \gamma \delta \right) },\ \ \ \ s\leq 0  \label{sup-GIG trawl} \\
leb\left( A\right) &=&\left( 1-b\right) \dfrac{\gamma }{\delta }\dfrac{%
K_{\nu -1}\left( \gamma \delta \right) }{K_{\nu }\left( \gamma \delta
\right) },  \notag \\
leb(A_{t}\cap A) &=&\left( 1-b\right) \dfrac{\gamma }{\delta }\dfrac{\left(
1+2t/\gamma ^{2}\right) ^{\left( 1-\nu \right) /2}K_{\nu -1}\left( \gamma
\delta \sqrt{1+2t/\gamma ^{2}}\right) }{K_{\nu }\left( \gamma \delta \right)
},\ \ \ \ t\geq 0,  \notag
\end{eqnarray}%
where $K_{\nu }\left( x\right) $ is the modified Bessel function of the 2nd
kind. It always has short memory as%
\begin{equation*}
\int_{0}^{\infty }\lambda ^{-2}\pi \left( \mathrm{d}\lambda \right) =\dfrac{%
\left( \gamma /\delta \right) ^{\nu }}{2K_{\nu }\left( \gamma \delta \right)
}\dfrac{2K_{\nu -2}\left( \gamma \delta \right) }{\left( \gamma /\delta
\right) ^{\nu -2}}=\left( \dfrac{\gamma }{\delta }\right) ^{2}\dfrac{K_{\nu
-2}\left( \gamma \delta \right) }{K_{\nu }\left( \gamma \delta \right) }%
<\infty \text{ for all }\gamma ,\delta >0,\ \nu \in
\mathbb{R}
.
\end{equation*}%
However, it can also degenerate to the long memory sup-$\Gamma $ trawl by
letting $\gamma =\sqrt{2\alpha }$, $\nu =H$ and $\delta \rightarrow 0$. When
$\gamma \rightarrow 0$, $\pi \left( \mathrm{d}\lambda \right) $ becomes an
inverse gamma distribution with scale parameter $\delta ^{2}/2$ and shape
parameter $-\nu $, so correspondingly we produce the superposition inverse
gamma (sup-$\Gamma ^{-1}$) trawl. This is an important case, for inverse
gamma densities have polynomial decay in their tails so will generate short
but substantial memory, which has the same pattern as the empirical data. We
will see this clearly in Section \ref{sect:futures data}.
\end{example}

\section{Moment-based inference\label{sect:moment based inference}}

Here we discuss the inference technique based on matching moments using a
path of prices $P_{t}$, $t\in \lbrack 0,T]$. Due to (i) the stationarity of
the price changes $P_{\delta }-P_{0}\overset{d}{\backsim }P_{t+\delta }-P_{t}
$ for any $t,\delta $ and (ii) the high frequency nature of the data,
moment-based estimates are plausible. The inference can basically split in
two pieces: the inference of the L\'{e}vy measure $\nu $ and the inference
on $b$ and $d$.

\subsection{Inference of L\'{e}vy measure}

Due to the high frequency nature of the data, the instantaneous jumping
distribution of the sample is close to the true value. Similarly, the sample
power variation $\left\{ P\right\} _{t}^{\left[ r\right] }$ for any $r\geq 0$%
, when treated as a linear function of time $t$, has a slope that is also
close to the truth. We can then use these facts to estimate the L\'{e}vy
measure $\nu $ in terms of $b$.

Let us write the sample instantaneous jumping distribution as $\hat{\alpha}%
_{y}$, where $\sum_{y\in
\mathbb{Z}
\backslash \left\{ 0\right\} }\hat{\alpha}_{y}=1$; also, estimate the slope
of the $r$-th sample power variation against $t$ by%
\begin{equation*}
\hat{\beta}_{r}\triangleq \dfrac{\left\{ P\right\} _{T}^{[r]}}{T}=\dfrac{1}{T%
}\sum_{0<t\leq T}\left\vert \Delta P_{t}\right\vert ^{r}.
\end{equation*}%
Then by matching moments to equations (\ref{Jumping distribution}) and (\ref%
{Expectation of the power variation}), we should have%
\begin{eqnarray}
\left( 2-b\right) \sum_{y\in
\mathbb{Z}
\backslash \left\{ 0\right\} }\left\vert y\right\vert ^{r}\nu \left(
y\right) &=&\hat{\beta}_{r},\ \ \ \ r\geq 0,
\label{Moment equation from power variation} \\
\nu \left( y\right) +\nu \left( -y\right) \left( 1-b\right) &=&\left(
2-b\right) \hat{\alpha}_{y}\left\Vert \nu \right\Vert ,\ \ \ \ y\in
\mathbb{Z}
\backslash \left\{ 0\right\} .
\label{Moment equation from jumping distribution}
\end{eqnarray}%
Using (\ref{Moment equation from power variation}) with the case of $r=0$,
we have $\left\Vert \nu \right\Vert =\sum_{y\in
\mathbb{Z}
\backslash \left\{ 0\right\} }\nu \left( y\right) =\hat{\beta}_{0}/\left(
2-b\right) $ and hence%
\begin{eqnarray*}
\nu \left( y\right) +\nu \left( -y\right) \left( 1-b\right) &=&\hat{\alpha}%
_{y}\hat{\beta}_{0}, \\
\nu \left( -y\right) +\nu \left( y\right) \left( 1-b\right) &=&\hat{\alpha}%
_{-y}\hat{\beta}_{0},\ \ \ \ y\in
\mathbb{N}
.
\end{eqnarray*}%
Solving these two equations gives us%
\begin{equation}
\widehat{\nu \left( y\right) }\triangleq \dfrac{\hat{\alpha}_{y}-\left(
1-b\right) \hat{\alpha}_{-y}}{\left( 2-b\right) b}\hat{\beta}_{0},\ \ \ \
y\in
\mathbb{Z}
\backslash \left\{ 0\right\} .
\label{Levy measure moment estimate in terms of b}
\end{equation}

\begin{remark}
This does not guarantee that $\widehat{\nu \left( y\right) }\geq 0$, so
empirically we will truncate negative $\widehat{\nu \left( y\right) }$ by
zero and at the same time tune the value of the corresponding $\widehat{\nu
\left( -y\right) }$ such that%
\begin{equation*}
\widehat{\nu \left( y\right) }+\widehat{\nu \left( -y\right) }=\dfrac{\hat{%
\alpha}_{y}-\left( 1-b\right) \hat{\alpha}_{-y}+\hat{\alpha}_{-y}-\left(
1-b\right) \hat{\alpha}_{y}}{\left( 2-b\right) b}\hat{\beta}_{0}=\dfrac{\hat{%
\alpha}_{y}+\hat{\alpha}_{-y}}{\left( 2-b\right) }\hat{\beta}_{0}
\end{equation*}%
remains unchanged. The advantage of this modification allows the
conservation of all the (non-negative) moments of the estimated L\'{e}vy
measure $\hat{\nu}$:%
\begin{equation*}
\sum_{y\in
\mathbb{Z}
\backslash \left\{ 0\right\} }\left\vert y\right\vert ^{r}\widehat{\nu
\left( y\right) }=\sum_{y=1}^{\infty }\left\vert y\right\vert ^{r}\left(
\widehat{\nu \left( y\right) }+\widehat{\nu \left( -y\right) }\right) .
\end{equation*}%
However, it comes with the price that the estimates for all of the \emph{odd}
cumulants of $P_{t}-P_{0}$ are altered, but practically this will be
neglectable as the truncation is only needed for larger $y$ and the
corresponding intensity $\nu \left( y\right) $ is usually quite small.

To completely avoid the negative estimates, one might parameterize the L\'{e}%
vy measure as in \cite{BarndorffNielsenLundeShephardVeraart(14)}, but here
we prefer to stay with the non-parametric estimates.
\end{remark}

\begin{remark}
\label{Rmk: b cannot be estimated seperately}We should note that (\ref{Levy
measure moment estimate in terms of b}) has included all the information we
can access from equations (\ref{Moment equation from power variation}) and (%
\ref{Moment equation from jumping distribution}), so we cannot rely on
equations (\ref{Moment equation from power variation}) and (\ref{Moment
equation from jumping distribution}) to solve $b$ and the L\'{e}vy measure $%
\nu $ at the same time. The details can be found in the Appendix \ref{Sec.:
Proof}.
\end{remark}

\subsection{Inference of permanence and trawl function}

We will need to employ additional moment equations to estimate the trawl
function $d$ as well as $b$. The easiest way to do this is through Theorem %
\ref{thm:dist of returns}. In particular, we will use the sample variance of
$\left\{ P_{k\delta }-P_{\left( k-1\right) \delta }\right\} _{k=1}^{T/\delta
}$ to estimate%
\begin{equation*}
\func{Var}\left( P_{\delta }-P_{0}\right) =\left( b\delta +2leb\left(
A_{\delta }\backslash A_{0}\right) \right) \kappa _{2}\left( L_{1}\right)
=\left( b\delta +2leb\left( A_{\delta }\backslash A_{0}\right) \right)
\sum_{y\in
\mathbb{Z}
\backslash \left\{ 0\right\} }y^{2}\nu \left( y\right) .
\end{equation*}

Denote the sample variance with the sampling interval $\delta $ as $\widehat{%
\sigma _{\delta }^{2}}$. Then by (\ref{Levy measure moment estimate in terms
of b}) and matching moments, we should have%
\begin{equation}
\widehat{\sigma _{\delta }^{2}}=\left( \dfrac{b\delta +2leb\left( A_{\delta
}\backslash A_{0}\right) }{2-b}\right) \sum_{y\in
\mathbb{Z}
\backslash \left\{ 0\right\} }y^{2}\hat{\alpha}_{y}\hat{\beta}_{0}.
\label{Variance signature}
\end{equation}%
Appendix \ref{Sec.: Nonparametric d} shows how to non-parametrically
estimate the trawl function $d$ using $\widehat{\sigma _{\delta }^{2}}$, but
here we only demonstrate the inference for a parameterized trawl.

Suppose for now that the trawl function $d$ is parameterized by $\phi $, for
example, $\phi =\lambda $ in the exponential trawl (\ref{OU trawl}), $\phi
=\left( \alpha ,H\right) ^{T}$ in the sup-$\Gamma $ trawl (\ref{long memory
trawl}) and $\phi =\left( \gamma ,\delta ,\nu \right) ^{T}$ in the sup-GIG
trawl (\ref{sup-GIG trawl}). A simple way to estimate $b$ and $\phi $
simultaneously is through a non-linear least square fitting to equation (\ref%
{Variance signature}) divided by $\delta $ across different $\delta $. The
reason to work on $\widehat{\sigma _{\delta }^{2}}/\delta $ instead of $%
\widehat{\sigma _{\delta }^{2}}$ is to amplify the effect of empirical
market microstructure for small $\delta $, so the non-linear least square
estimation of $b$ and $\phi $ will not be overly dominated by the linear
part of the variogram.

\begin{remark}
\label{Rmk.: Variant of the variance signature plot}By definition of the
sample variance and the realized variance, as $T\rightarrow \infty $,%
\begin{eqnarray*}
\widehat{\sigma _{\delta }^{2}} &\approx &\dfrac{1}{T/\delta }%
\sum_{k=1}^{T/\delta }\left( P_{k\delta }-P_{\left( k-1\right) \delta
}\right) ^{2}-\left( \dfrac{P_{T}-P_{0}}{T/\delta }\right) ^{2}, \\
\dfrac{\widehat{\sigma _{\delta }^{2}}}{\delta } &\approx &\dfrac{1}{T}%
RV^{\left( T/\delta \right) }-\delta \dfrac{\left( P_{T}-P_{0}\right) ^{2}}{%
T^{2}}\approx \dfrac{1}{T}RV^{\left( T/\delta \right) },
\end{eqnarray*}%
where we throw out the second-order term in the final approximation. Thus,
essentially what we try to fit is the variance signature plot ($RV^{\left(
T/\delta \right) }$ against $\delta $). From now on, we also call the plot $%
\widehat{\sigma _{\delta }^{2}}/\delta $ against $\delta $ a variance
signature plot.
\end{remark}

\begin{example}
To check the effectiveness of this moment estimator, we conduct a Monte
Carlo simulation study on the price process model parameterized by the
exponential trawl (\ref{OU trawl}). Throughout the rest of this paper, all
the numerical values are reported under the time unit being a second. Then
we set $\lambda _{\mathrm{true}}=0.681$ and a non-symmetric Skellam basis
with L\'{e}vy measure%
\begin{equation*}
\nu \left( \mathrm{d}y\right) =\nu ^{+}\delta _{\left\{ 1\right\} }\left(
\mathrm{d}y\right) +\nu ^{-}\delta _{\left\{ -1\right\} }\left( \mathrm{d}%
y\right) ,
\end{equation*}%
where $\nu _{\mathrm{true}}^{+}=0.0138$, $\nu _{\mathrm{true}}^{-}=0.0131$
and $b_{\mathrm{true}}=0.396$. All the $10,000$ Monte Carlo simulated paths
are drawn with $V_{0}=7,486$ (ticks) during the time interval $72.03$ to $%
75,600$ (seconds), where $75,600$ means the closing time of the market,
21:00. All the settings here are taken from the empirical TNC1006 data set
on March 22, 2010, which we will study in next Section.

The non-linear least square fitting for (\ref{Variance signature}) is
conducted for $\delta $'s ranging from $0.1$ seconds to $60$ seconds with $%
60 $ equally spaced grid points on its log-scale. We then repeat the
moment-based estimates for $\theta =\left( b,\nu ^{+},\nu ^{-},\lambda
\right) ^{T}$ and derive histograms of these estimates in Figure \ref{fig:
MonteCarlo MomentEst for exp-skellam1}.
\begin{figure}[t]
\label{fig: MonteCarlo MomentEst for exp-skellam}\centering
\includegraphics{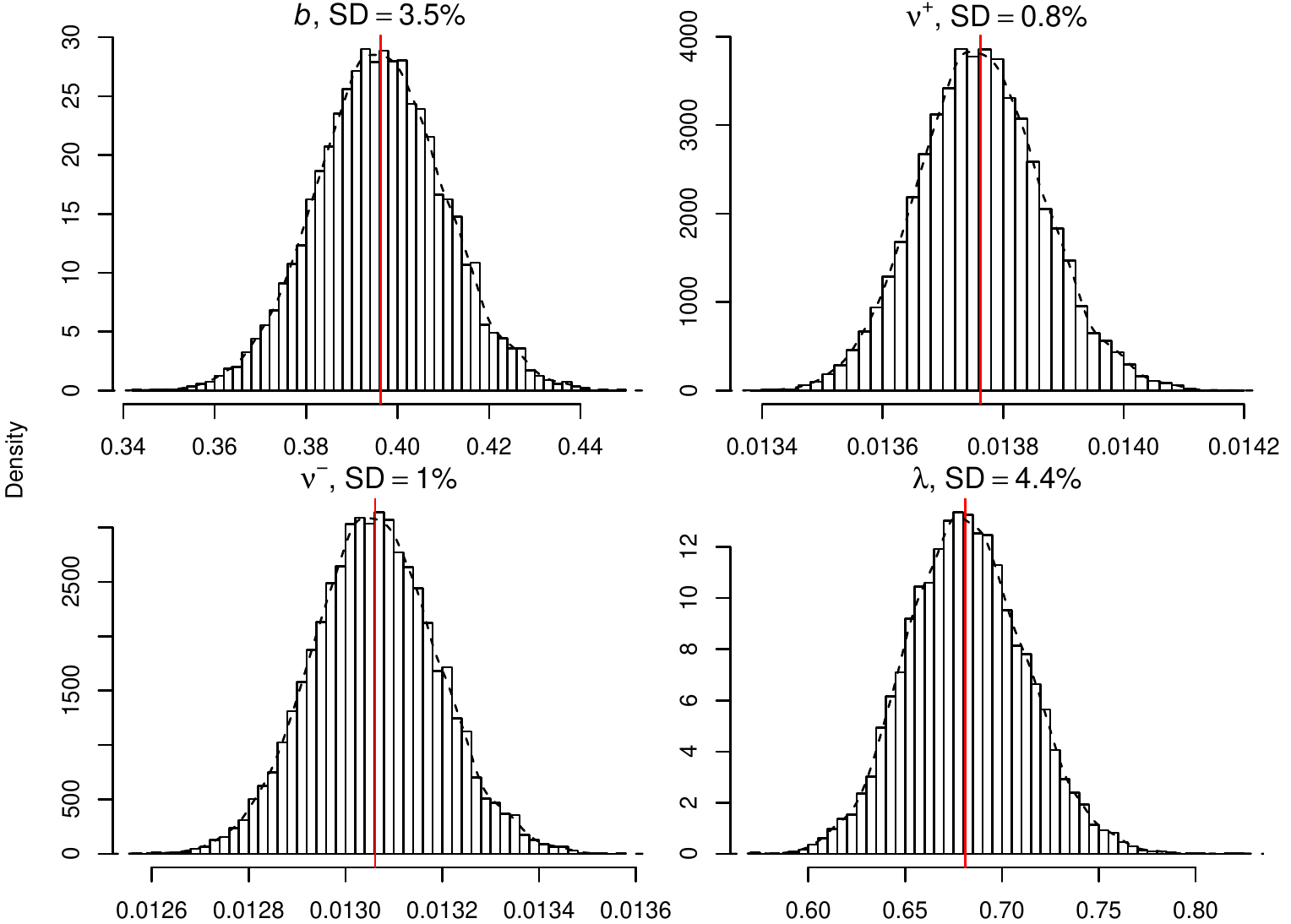}
\caption{$10,000$ Monte Carlo simulation of moment estimations for the price
process with exponential trawl $d\left( s\right) =b+\left( 1-b\right) \exp
\left( \protect\lambda s\right) $ and the Skellam basis $\protect\nu \left(
\mathrm{d}y\right) =\protect\nu ^{+}\protect\delta _{\left\{ 1\right\}
}\left( \mathrm{d}y\right) +\protect\nu ^{-}\protect\delta _{\left\{
-1\right\} }\left( \mathrm{d}y\right) $. The vertical lines in each of the
histograms mean the true value. The Monte Carlo standard deviations are
reported on the scale of the true values. Code: \texttt{Moment\_Inference%
\_ModelBasedBootstrap.R}.}
\label{fig: MonteCarlo MomentEst for exp-skellam1}
\end{figure}
The estimates from the proposed methodology (using equations (\ref{Levy
measure moment estimate in terms of b}) and (\ref{Variance signature}))
correctly center around the true values; also notice that this method is
particularly accurate for estimating $\nu ^{+}$ and $\nu ^{-}$.
\end{example}

\begin{remark}
Except the moment-based estimations, we can also conduct maximum likelihood
estimation for our proposed model, which requires sophisticated techniques
to filter out the L\'{e}vy process $L\left( B_{t}\right) $. We are currently
exploring particle methods toward this direction.
\end{remark}

\section{Empirical analysis for futures data\label{sect:futures data}}

In this Section, we employ these moment-based estimators for empirical
analysis. Covering two days of trading activities on two different assets,
four data sets are studied here: (i) the Ten-Year US Treasury Note futures
contract delivered in June 2010 (TNC1006) during March 22, 2010; (ii) the
International Monetary Market (IMM) Euro-Dollar Foreign Exchange (FX)
futures contract delivered in June 2010 (EUC1006) during March 22, 2010;
(iii) TNC1006 during May 7, 2010; (iv) EUC1006 during May 7, 2010. These
data sets come from the same database that is used by \cite%
{BarndorffNielsenPollardShephard(12)}. The first trading day is randomly
chosen, while the second trading day is not only the release of US non-farm
payroll numbers but was also experiencing the European sovereign debt
crisis. These data sets are derived from data feeds at the Chicago
Mercantile Exchange (CME). They have been preprocessed using the procedures
described in Appendix \ref{Sec.: Data cleaning}. From now on, we will no
longer mention the delivery date of each data set and the year 2010.

\subsection{Data features}

All of these four data sets use all the trades from 00:00 to 21:00, shown in
Figure \ref{fig.: Empirical Long trace plot}.
\begin{figure}[t]
\centering\includegraphics{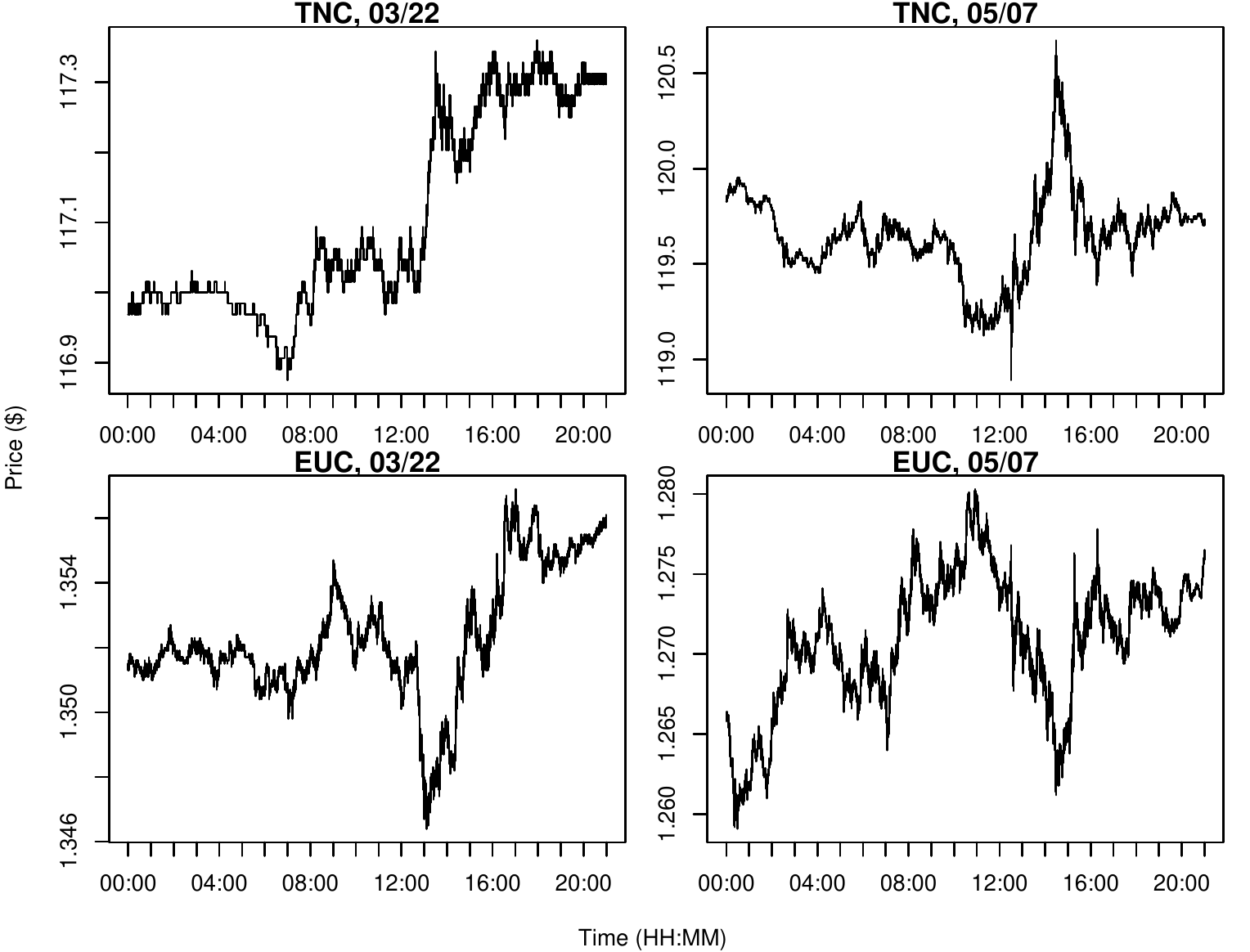}
\caption{The complete trace plots for the four data sets during 00:00 to
21:00. The $x$-axis is the calendar time (HH:MM), while the $y$-axis is the
price (\$). Code: \texttt{Price\_Plots.R}.}
\label{fig.: Empirical Long trace plot}
\end{figure}
With such large time scales, each of the trace plots look like a continuous
time diffusion process. However, if we focus these data sets to much smaller
time scales (within one hour for TNC and within two minutes for EUC), shown
in Figure \ref{fig.: Empirical Short trace plot},
\begin{figure}[t]
\centering\includegraphics{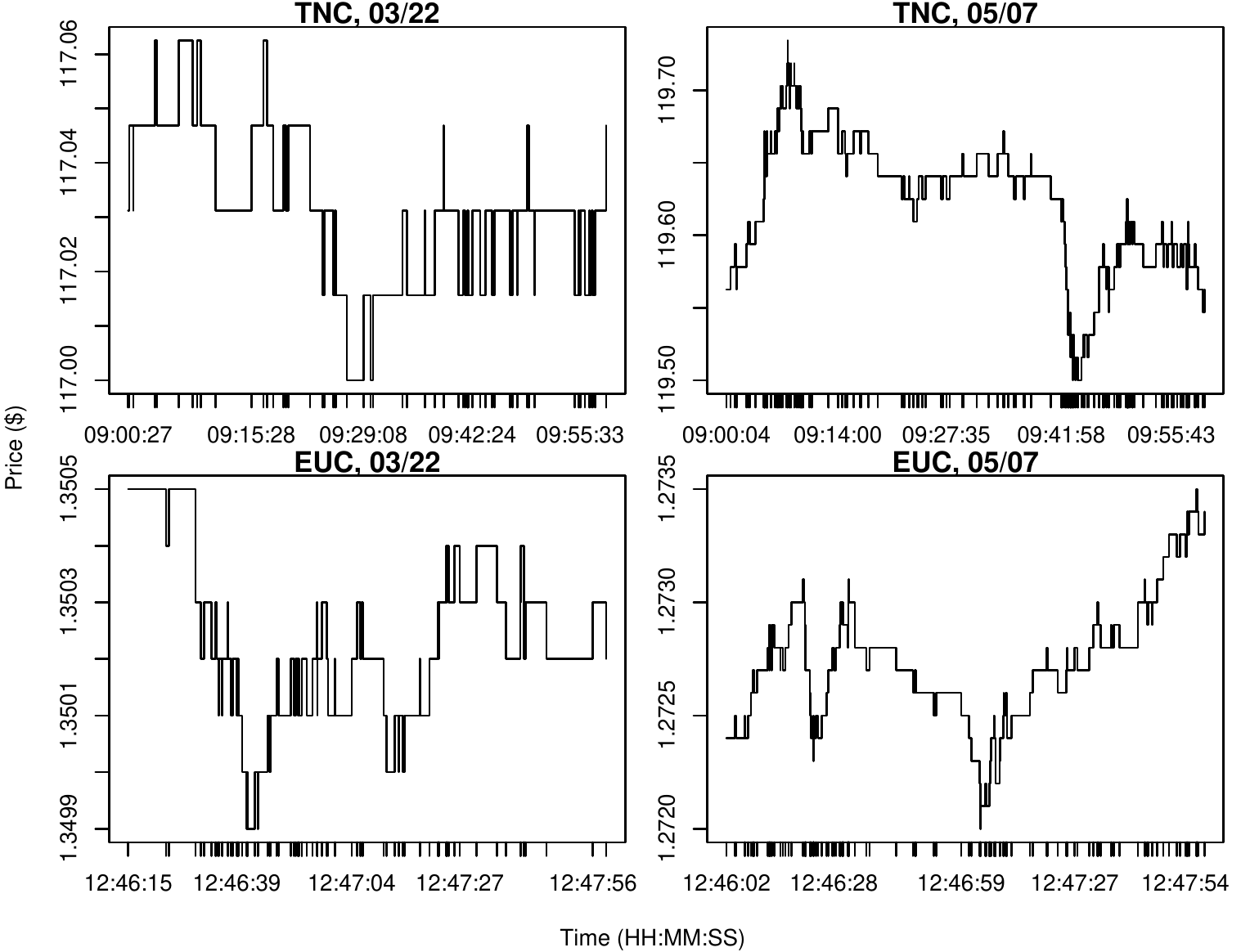}
\caption{The trace plots for two TNC data sets during 09:00 to 10:00 and for
two EUC data sets during 12:46 to 12:48. The $x$-axis is the calendar time
(HH:MM:SS), while the $y$-axis is the price (\$). Code: \texttt{%
Price\_Plots.R}.}
\label{fig.: Empirical Short trace plot}
\end{figure}
the discreteness becomes important. In particular, we can see several
multiple-tick jumps in the two EUC data sets shown in Figure \ref{fig.:
Empirical Short trace plot}.

Table \ref{Tab.: Summary stat of the 4 data set} summarizes some basic
features of these four data sets.
\begin{table}[tbh] \centering%
$%
\begin{tabular}{|r|r|r|r|r|r|r|}
\hline
\multirow{2}{*}{Contract, Day} & \multirow{2}{*}{Tick Size (\$)} & %
\multirow{2}{*}{Num. of Price Changes} & \multicolumn{4}{|c|}{Size of Price
Changes (Tick)} \\ \cline{4-7}
&  &  & Avg. & SD. & Min. & Max. \\ \hline
TNC,\ 03/22 & $1/64$ & $3,249$ & $0.00646$ & $1.000$ & $-1$ & $1$ \\ \hline
EUC, 03/22 & $0.0001$ & $13,943$ & $0.00337$ & $1.012$ & $-2$ & $3$ \\ \hline
TNC, 05/07 & $1/64$ & $12,849$ & $-0.000467$ & $1.035$ & $-13$ & $15$ \\
\hline
EUC, 05/07 & $0.0001$ & $55,379$ & $0.00190$ & $1.077$ & $-13$ & $15$ \\
\hline
\end{tabular}%
$\caption{Summary statistics of the four futures data sets.}\label{Tab.:
Summary stat of the 4 data set}%
\end{table}
Both contracts have more activities during May 7 than during March 22 and
the standard deviations of the jump size for all the four data sets are
close to $1$ even though the range of all possible jump sizes might differ a
lot.

We also plot the empirical instantaneous jumping distribution (on the
log-scale) for the four data sets in Figure \ref{fig.:Empirical Jumping
Distribution}. Those estimated probabilities will be used as $\hat{\alpha}%
_{y}$ for the moment estimate defined in the previous Section.
\begin{figure}[t]
\centering\includegraphics{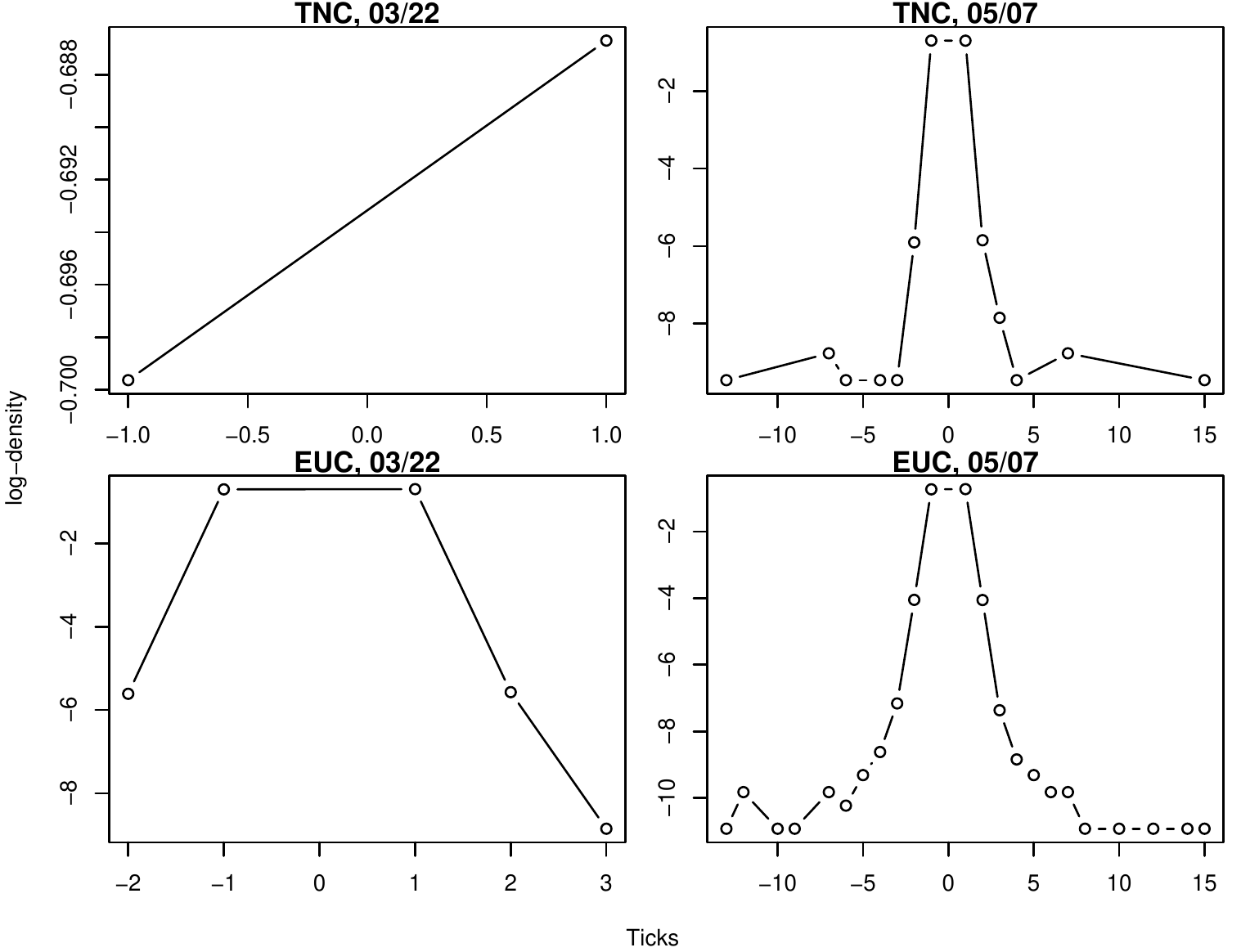}
\caption{The log-histograms for the empirical instantaneous jumping
distributions of the four data sets. The $x$-axis for each plot is the size
of the jump, while the $y$-axis denotes the estimated probability value in a
log-scale. Code: \texttt{Price\_Plots.R}.}
\label{fig.:Empirical Jumping Distribution}
\end{figure}
Generally, the jumps of EUC have more variability than the TNC. Furthermore,
we can see that even for the same contract, say TNC, the jumping
characteristic is completely different from a random chosen day (March 22)
to a day with a major economic event (May 7). In a normal day like March 22,
the TNC trading has depths so large that it always jumps by one tick, but
the situation changes enormously for a highly active day like May 7, by this
time the TNC trading behaves just like other multiple-tick markets.

\begin{remark}
\label{Rmk.: The small kappa1}One more implication from Figure \ref%
{fig.:Empirical Jumping Distribution} is that $\kappa _{1}\left(
L_{1}\right) $ is, of course, a small number. To see this, we note that%
\begin{equation*}
\widehat{\kappa _{1}\left( L_{1}\right) }=\sum_{y\in
\mathbb{Z}
\backslash \left\{ 0\right\} }y\widehat{\nu \left( y\right) }=\dfrac{%
\sum_{y\in
\mathbb{Z}
\backslash \left\{ 0\right\} }y\hat{\alpha}_{y}-\left( 1-b\right) \sum_{y\in
\mathbb{Z}
\backslash \left\{ 0\right\} }y\hat{\alpha}_{-y}}{\left( 2-b\right) b}\hat{%
\beta}_{0}=\dfrac{\sum_{y\in
\mathbb{Z}
\backslash \left\{ 0\right\} }y\hat{\alpha}_{y}}{b}\hat{\beta}_{0}.
\end{equation*}%
Hence, the more symmetric the Figure \ref{fig.:Empirical Jumping
Distribution}, the smaller the estimate of $\kappa _{1}\left( L_{1}\right) $.
\end{remark}

Finally, we show the correlograms of the four data sets in Figure \ref%
{fig.:Empirical Autocorrelation}, using three orders of magnitude of
sampling intervals $\delta $: 0.1 second, 1 second, 10 seconds and 1 minute.
For each data set, we will use a single set of parameters in our price model
to fit all of the correlograms with different $\delta $.
\begin{figure}[t]
\centering
\includegraphics{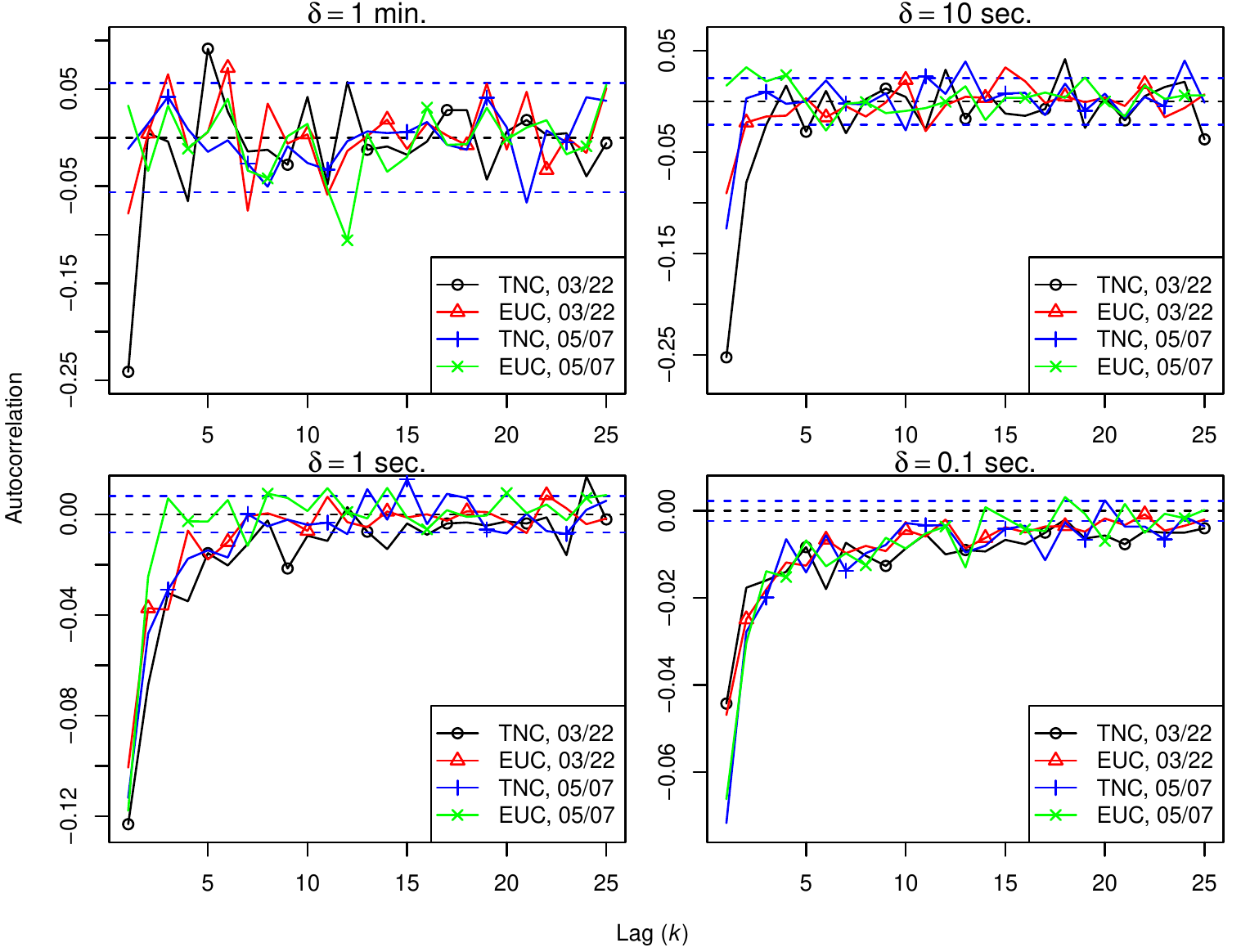}
\caption{The correlograms with different sampling intervals $\protect\delta %
=0.1,1,10,60$ (seconds) for the four data sets. The $x$-axis for each plot
is the lag $k$, while the $y$-axis denotes the value of empirical
autocorrelation. The dashed lines are located at $\pm 2/\protect\sqrt{T/%
\protect\delta }$. Code: \texttt{Price\_Plots.R}.}
\label{fig.:Empirical Autocorrelation}
\end{figure}
In general, these autocorrelations are significantly negative and increasing
as $k$ increases, while if $\delta $ gets very large the autocorrelations
will fall to roughly zero. Of course there is strong evidence that the
empirical data cannot be well-described by a pure L\'{e}vy process, which
always gives zero autocorrelations for returns. Our model is capable of
describing these autocorrelation features (Theorem \ref{Thm.: dependence of
returns} and Corollary \ref{Cor.: Negative Correlation}). The next
Subsection conducts moment-based estimations for these empirical data sets.

\subsection{Parameter estimation}

We use the methodology described before on the four data sets with the three
different trawls (\ref{OU trawl}), (\ref{long memory trawl}) and (\ref%
{sup-GIG trawl}). The estimation\footnote{%
To especially emphasize the fitting of market microstructure effects, the
sample variance is calculated on an equally distant grid on the log-scale of
$\delta $ whose range is shown in Figure \ref{fig.:Empirical Variance
Signature PlotLog}.} results are shown in Table
\vref{Tab.: Empirical Moment
Estimates}, where%
\begin{equation*}
\nu ^{+}\triangleq \sum_{y=1}^{\infty }\nu \left( y\right) \text{ and }\nu
^{-}\triangleq \sum_{y=1}^{\infty }\nu \left( -y\right)
\end{equation*}%
are the positive and negative jump intensities respectively.
\begin{table}[t]
\centering
\begin{tabular}{|c|c||c|c||c|c||c|c||c|c|}
\hline
\multirow{2}{*} {Trawl} & \multirow{2}{*} {Para} & \multicolumn{2}{|c||}{
TNC, 03/22} & \multicolumn{2}{|c||}{EUC, 03/22} & \multicolumn{2}{|c||}{TNC,
05/07} & \multicolumn{2}{|c|}{EUC, 05/07} \\ \cline{3-10}
&  & Est. & SE & Est. & SE & Est. & SE & Est. & SE \\ \hline
 \multirow{4}{*}{Exp} & $b$       & 0.396 & 0.014 & 0.654 & 0.008 & 0.574 & 0.015 & 0.694 & 0.007 \\
   & $\nu^+$                     & 0.014 & 0.000 & 0.069 & 0.000 & 0.059 & 0.001 & 0.282 & 0.001 \\
   & $\nu^-$                     & 0.013 & 0.000 & 0.068 & 0.000 & 0.060 & 0.001 & 0.279 & 0.001 \\
   & $\lambda$                    & 0.681 & 0.030 & 2.470 & 0.083 & 3.888 & 0.218 & 4.033 & 0.133 \\
   \hline
\multirow{5}{*}{sup-$\Gamma$}& $b$& 0.283 & 0.021 & 0.604 & 0.012 & 0.525 & 0.016 & 0.649 & 0.010 \\
   & $\nu^+$                     & 0.013 & 0.000 & 0.067 & 0.001 & 0.057 & 0.001 & 0.272 & 0.002 \\
   & $\nu^-$                     & 0.012 & 0.000 & 0.066 & 0.001 & 0.058 & 0.001 & 0.270 & 0.002 \\
   & $\alpha$                     & 1.146 & 0.191 & 0.311 & 0.037 & 0.187 & 0.038 & 0.192 & 0.023 \\
   & $H$                          & 1.000 & 0.125 & 1.000 & 0.104 & 1.000 & 0.139 & 1.000 & 0.102 \\
   \hline
\multirow{6}{*}{sup-GIG} & $b$    & 0.186 & 0.028 & 0.528 & 0.034 & 0.440 & 0.029 & 0.648 & 0.011 \\
   & $\nu^+$                     & 0.013 & 0.000 & 0.063 & 0.001 & 0.054 & 0.001 & 0.272 & 0.002 \\
   & $\nu^-$                     & 0.011 & 0.000 & 0.062 & 0.001 & 0.055 & 0.001 & 0.269 & 0.002 \\
   & $\gamma$                     & 0.000 & 0.066 & 0.000 & 0.030 & 0.003 & 0.028 & 0.000 & 0.064 \\
   & $\delta$                     & 0.453 & 0.049 & 0.604 & 0.085 & 0.583 & 0.099 & 1.525 & 0.209 \\
   & $\nu$                        & -0.604 & 0.078 & -0.453& 0.067 & $-$0.332 & 0.077 & $-$0.741 & 0.170\\
   \hline

\end{tabular}
\caption{ Moment-based estimations under different trawls for the four data
sets. Also shown are the standard error (SE) estimates for the moment
estimator to each parameter using the model-based bootstrap, where the
number of bootstrapped paths we draw is 10,000. }
\label{Tab.: Empirical Moment Estimates}
\end{table}
We observe in the Table that the estimation of $\nu ^{+}$ and $\nu ^{-}$ are
relatively robust across different choices of trawls. The estimate of $H$ in
Table \ref{Tab.: Empirical Moment Estimates} clearly suggests the
insufficiency of using a sup-$\Gamma $ trawl for the empirical data.
Furthermore, even though we fit a more general sup-GIG trawl with three
parameters, the four empirical data sets can almost be described by the sup-$%
\Gamma ^{-1}$ trawl with only two parameters (the case of $\gamma
\rightarrow 0$ for sup-GIG trawl mentioned in Section \ref%
{sect:parameterised trawl}). This phenomenon might be attributed to the fact
that inverse gamma distributions decay exponentially near the origin but
polynomially near infinity, allowing it to capture these very different time
scales.

\begin{remark}
In the same Table, we also provide the standard error (SE) estimates for
these moment-based estimations using the model-based bootstrap, i.e., a
vanilla Monte Carlo simulation with plugged-in parameters.
\end{remark}

Using these estimated parameters, we first show the variance signature plots
of $\widehat{\sigma _{\delta }^{2}}/\delta $ against $\delta $ along with
the corresponding theoretical curves (\ref{Variance signature}) for each
trawl in Figure \ref{fig.:Empirical Variance Signature Plot} and \ref%
{fig.:Empirical Variance Signature PlotLog}, where the second of these
graphs uses a log-scale for $\delta $.
\begin{figure}[t]
\centering\includegraphics[width=%
\textwidth]{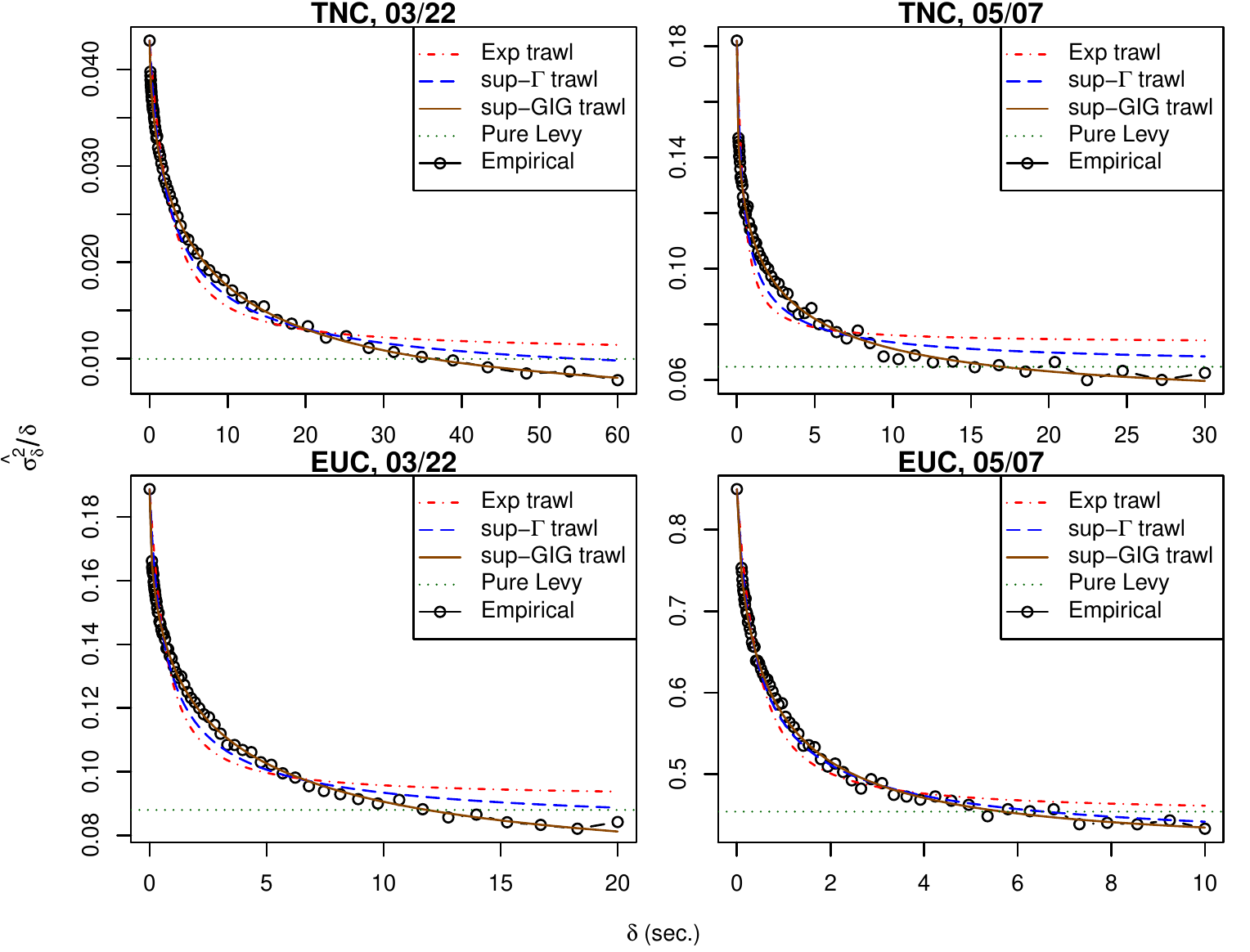}
\caption{The variance signature plots for the four data sets along with the
fitting curves from different trawls. The $x$-axis for each plot is $\protect%
\delta $ (seconds), while the $y$-axis denotes the value of the sample
variance of returns divided by $\protect\delta $. Code: \texttt{%
Moment\_Inference\_v2.0.R}.}
\label{fig.:Empirical Variance Signature Plot}
\end{figure}
In each of the plots, we put not only $\lim\limits_{\delta \rightarrow 0}%
\widehat{\sigma _{\delta }^{2}}/\delta =\left( \partial _{\delta }\widehat{%
\sigma _{\delta }^{2}}\right) \left( 0\right) =\sum_{y\in
\mathbb{Z}
\backslash \left\{ 0\right\} }y^{2}\hat{\alpha}_{y}\hat{\beta}_{0}$ at the
corresponding location of $\delta =0$ but also a reference horizontal line from a
pure L\'{e}vy process model ($b=1$), which is calculated from the slope of a
linear fitting line in the variogram of $\widehat{\sigma _{\delta }^{2}}$
against $\delta $.
\begin{figure}[t]
\centering\includegraphics[width=%
\textwidth]{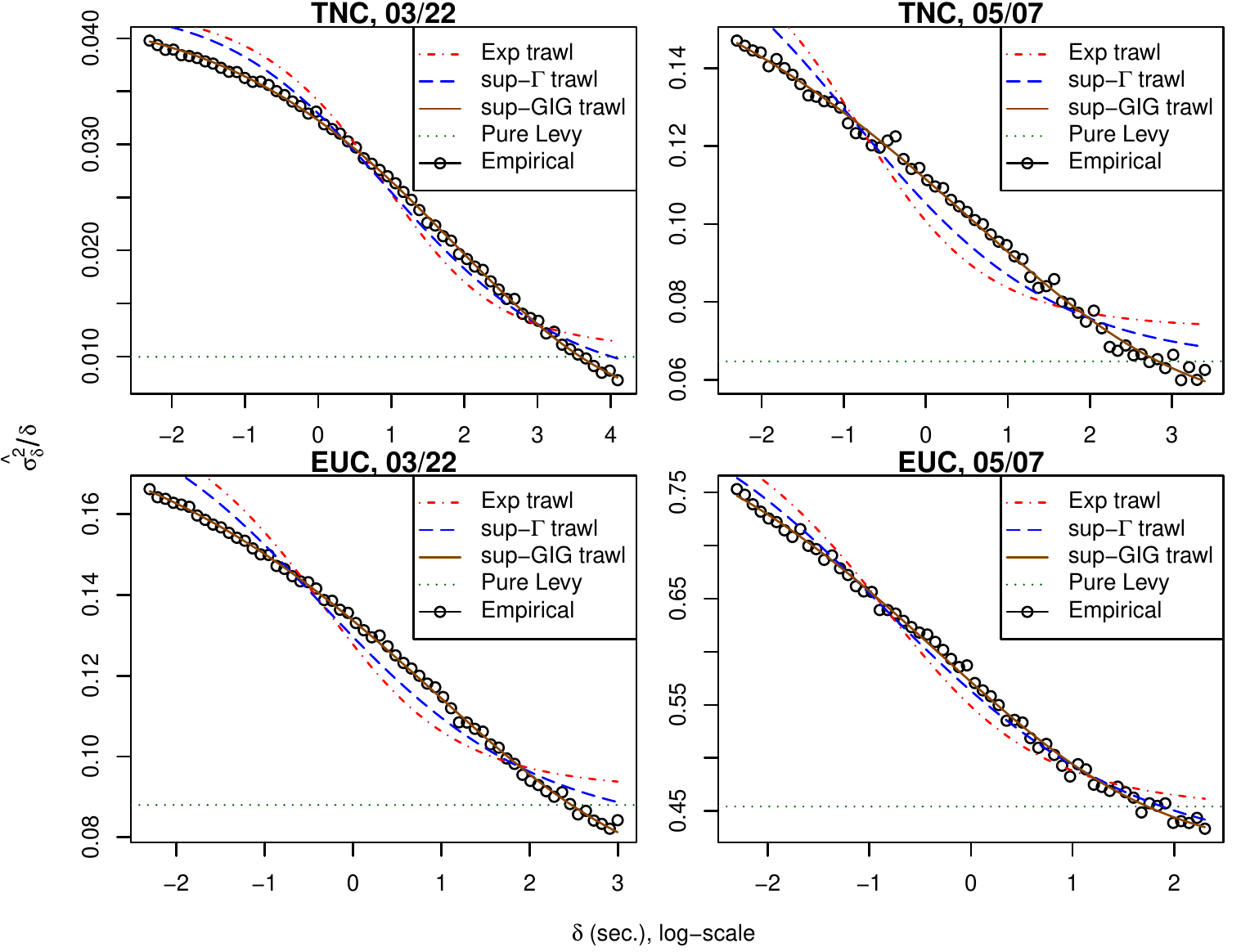}
\caption{The variance signature plots for the four data sets along with the
fitting curves from different trawls in the scale of $\log \protect\delta $.
Code: \texttt{Moment\_Inference\_v2.0.R}.}
\label{fig.:Empirical Variance Signature PlotLog}
\end{figure}

These fittings to the variance signature plots show good results---here we
particularly notice that using a sup-GIG trawl gives a very good fit; while
the other two simpler trawls fail to fit the region with a smaller $\delta $%
. This point becomes apparent when we check Figure \ref{fig.:Empirical
Variance Signature PlotLog}.

To further examine our model fitting, we also show the log-histograms for
the return distribution with different $\delta $ along with the theoretical
curves (by applying the inverse Fourier transform on Theorem \ref{thm:dist
of returns}) in Figure \ref{fig.:Empirical log-histograms of jumps}.
\begin{figure}[t]
\centering%
\begin{subfigure}[t]{0.48\textwidth}
        \centering
        \includegraphics[width=\textwidth]{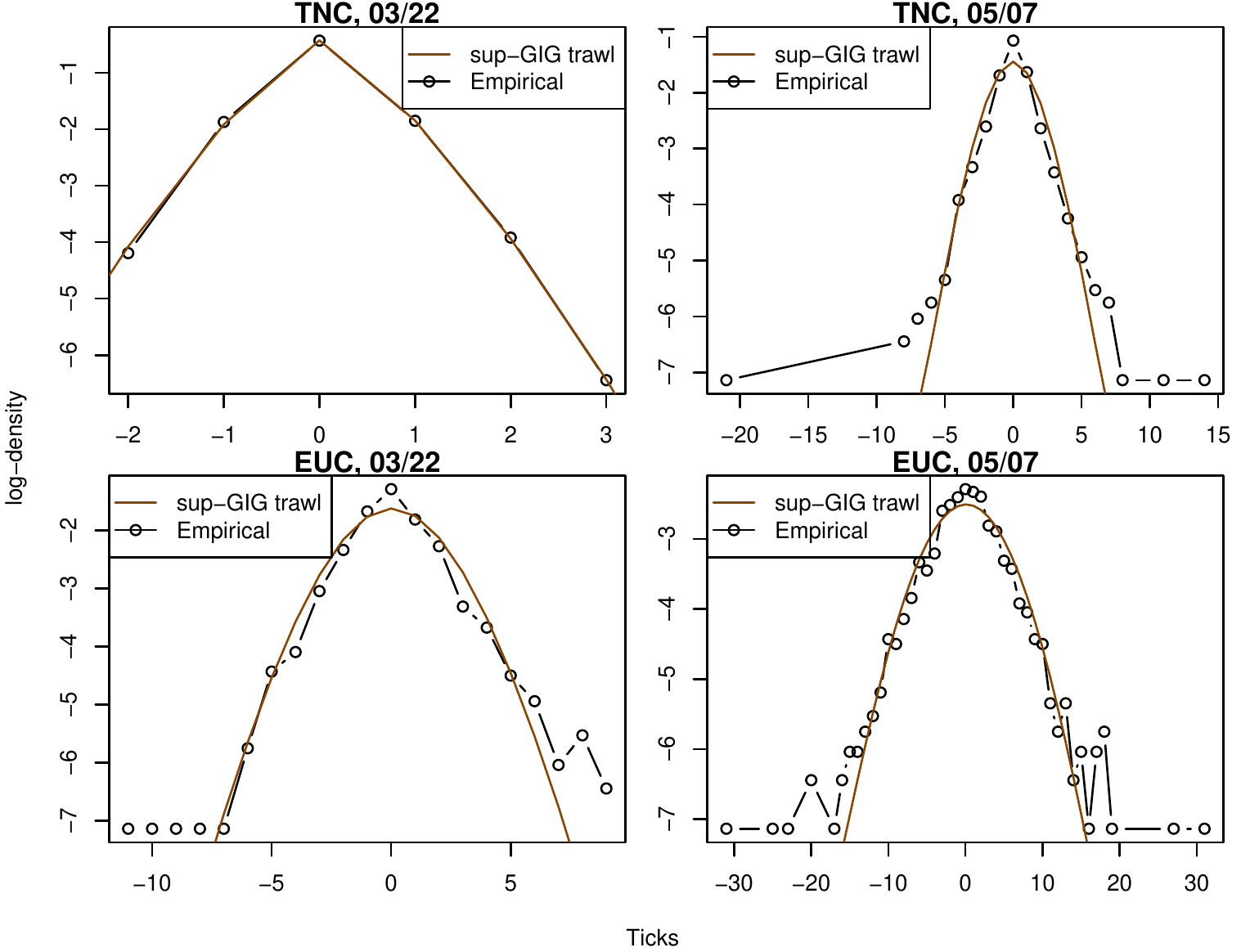}
        \caption{$\delta=60$ seconds.}
    \end{subfigure}%
\begin{subfigure}[t]{0.48\textwidth}
        \centering
        \includegraphics[width=\textwidth]{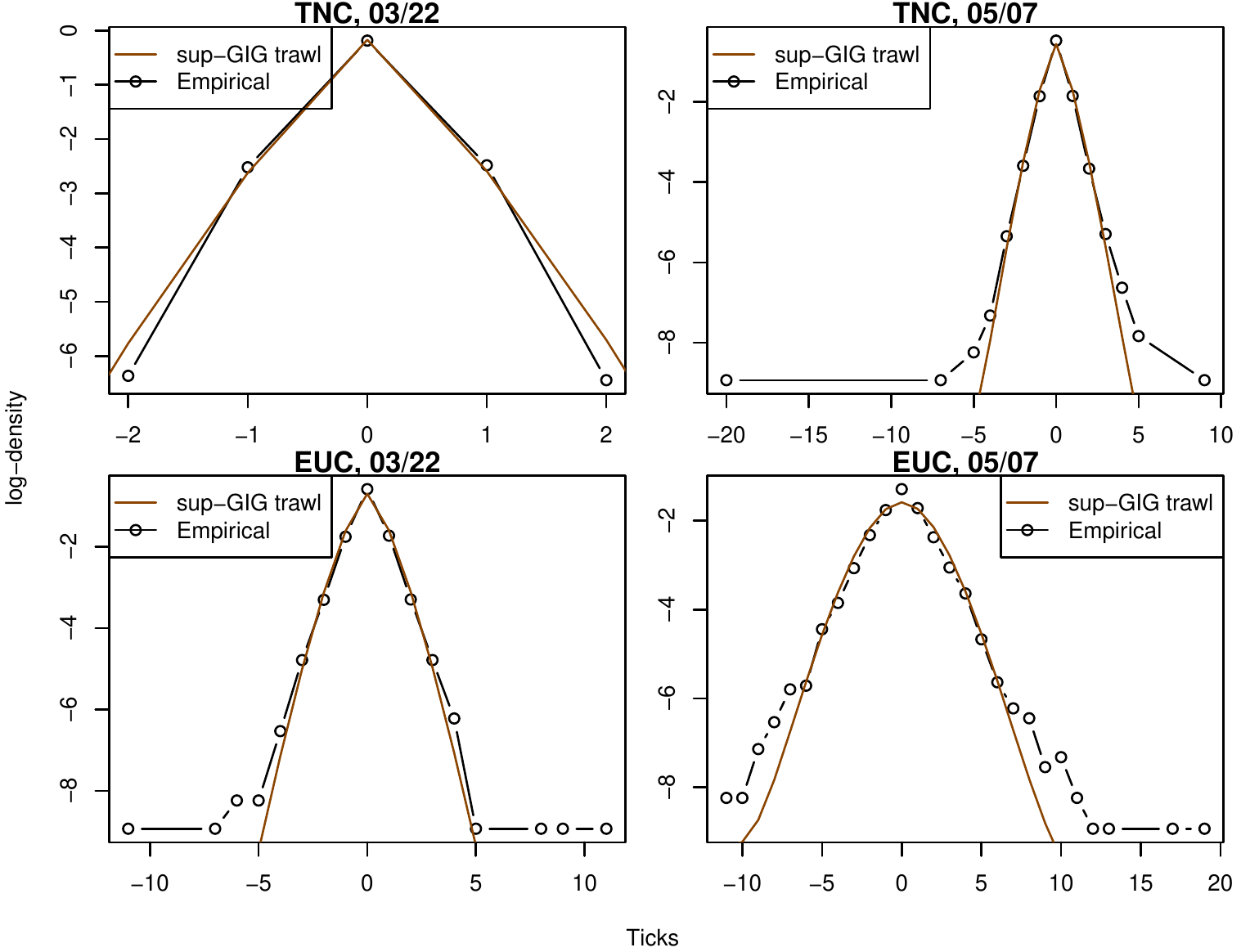}
        \caption{$\delta=10$ seconds.}
    \end{subfigure}\newline
\begin{subfigure}[t]{0.48\textwidth}
        \centering
        \includegraphics[width=\textwidth]{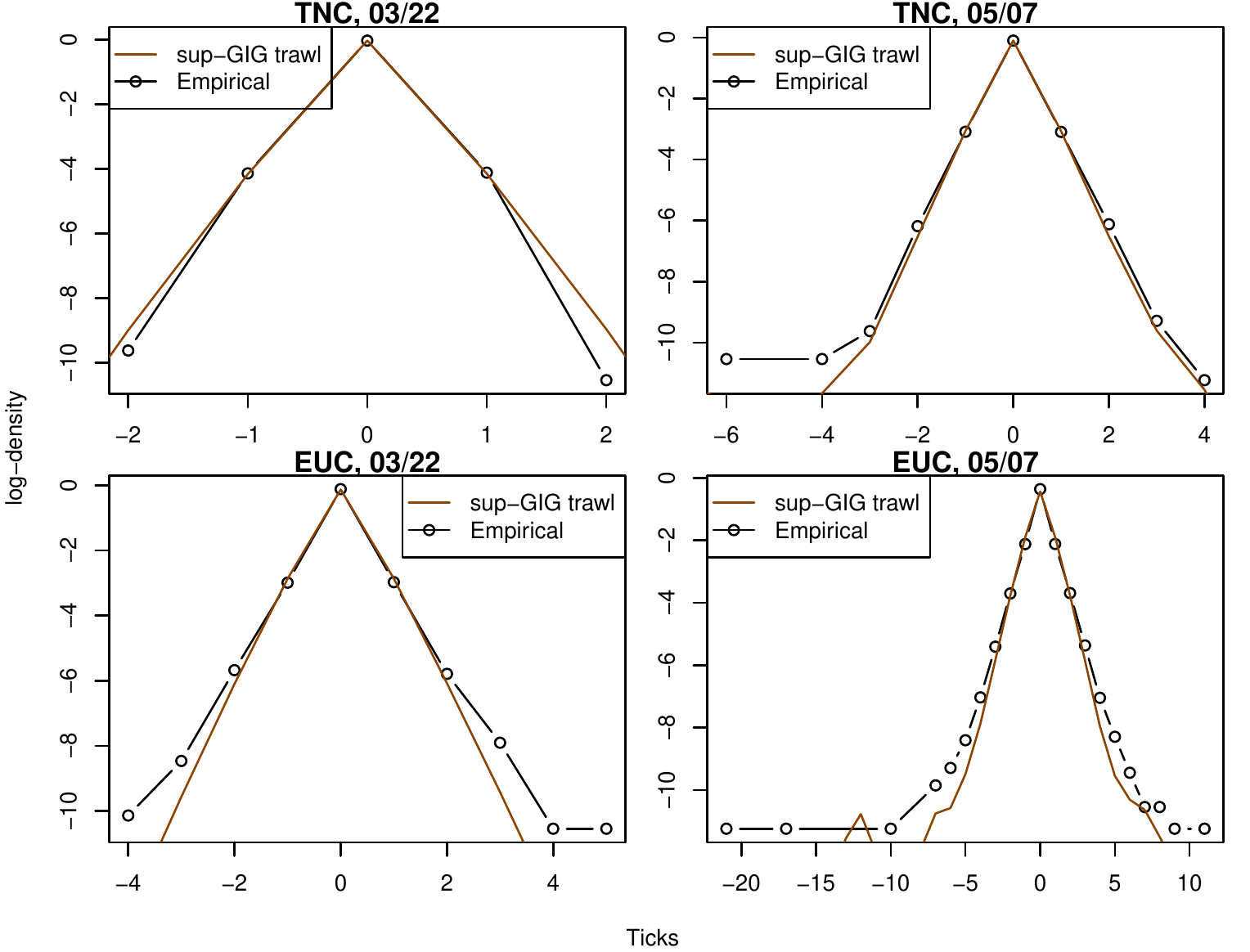}
        \caption{$\delta=1$ second.}
    \end{subfigure}%
\begin{subfigure}[t]{0.48\textwidth}
        \centering
        \includegraphics[width=\textwidth]{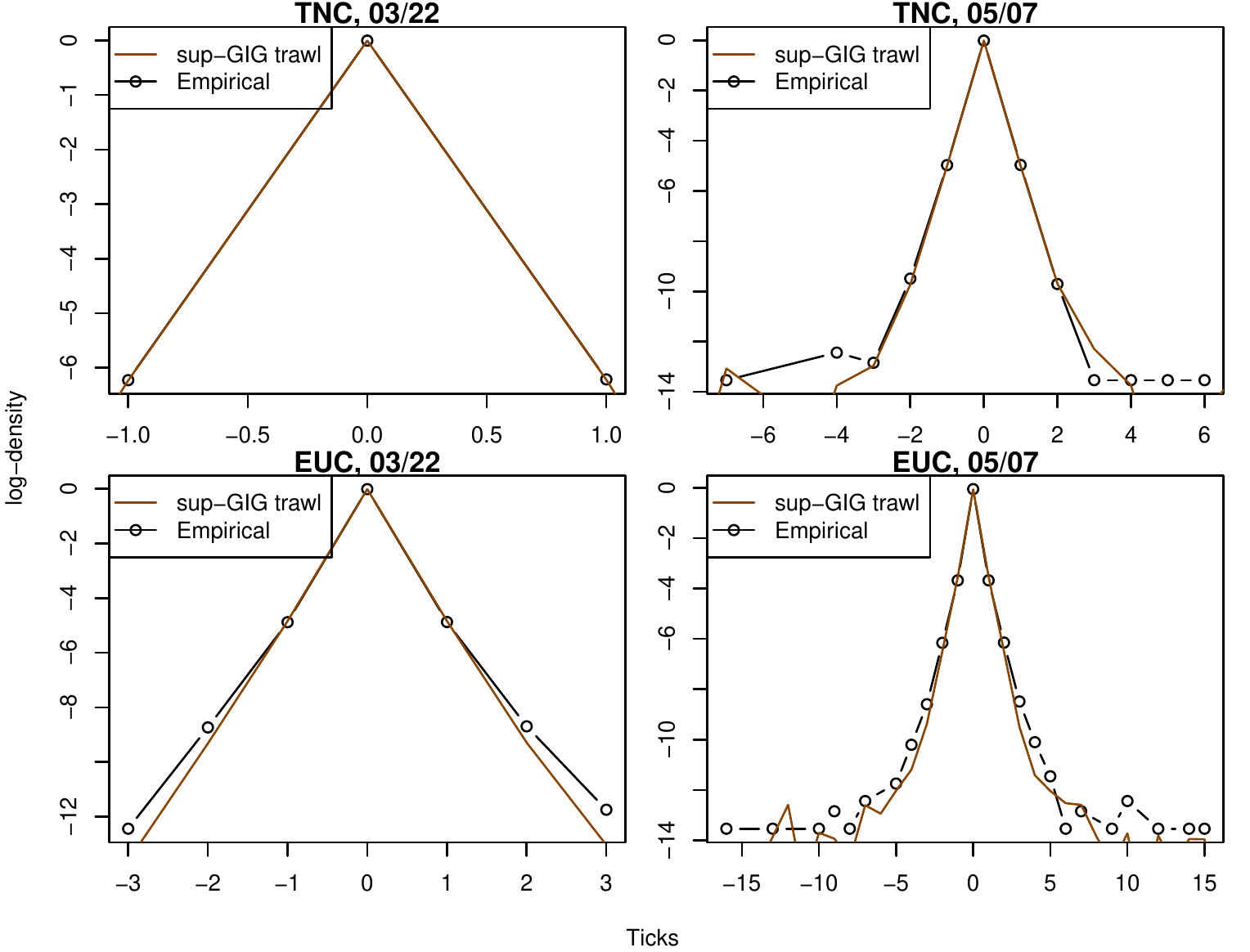}
        \caption{$\delta=0.1$ seconds.}
    \end{subfigure}
\caption{The log-histograms for the returns of the four data sets over
several sampling intervals along with the theoretical curves from sup-GIG
trawl. Code: \texttt{Moment\_Inference\_v2.0.R}.}
\label{fig.:Empirical log-histograms of jumps}
\end{figure}
For a larger $\delta $ the sup-GIG trawl do a better job than the other two
trawls (not shown in Figure \ref{fig.:Empirical log-histograms of jumps})
while for a smaller $\delta $ the difference among the three trawls is
limited. As an overall comment, our model seems to underestimate the tail
part of each of the empirical jumping distributions.

We now demonstrate the correlograms for the returns with different $\delta $
along with the theoretical curves in Figure \ref{fig.:Empirical ACF with
fitting}.
\begin{figure}[t]
\centering%
\begin{subfigure}[t]{0.48\textwidth}
        \centering
        \includegraphics[width=\textwidth]{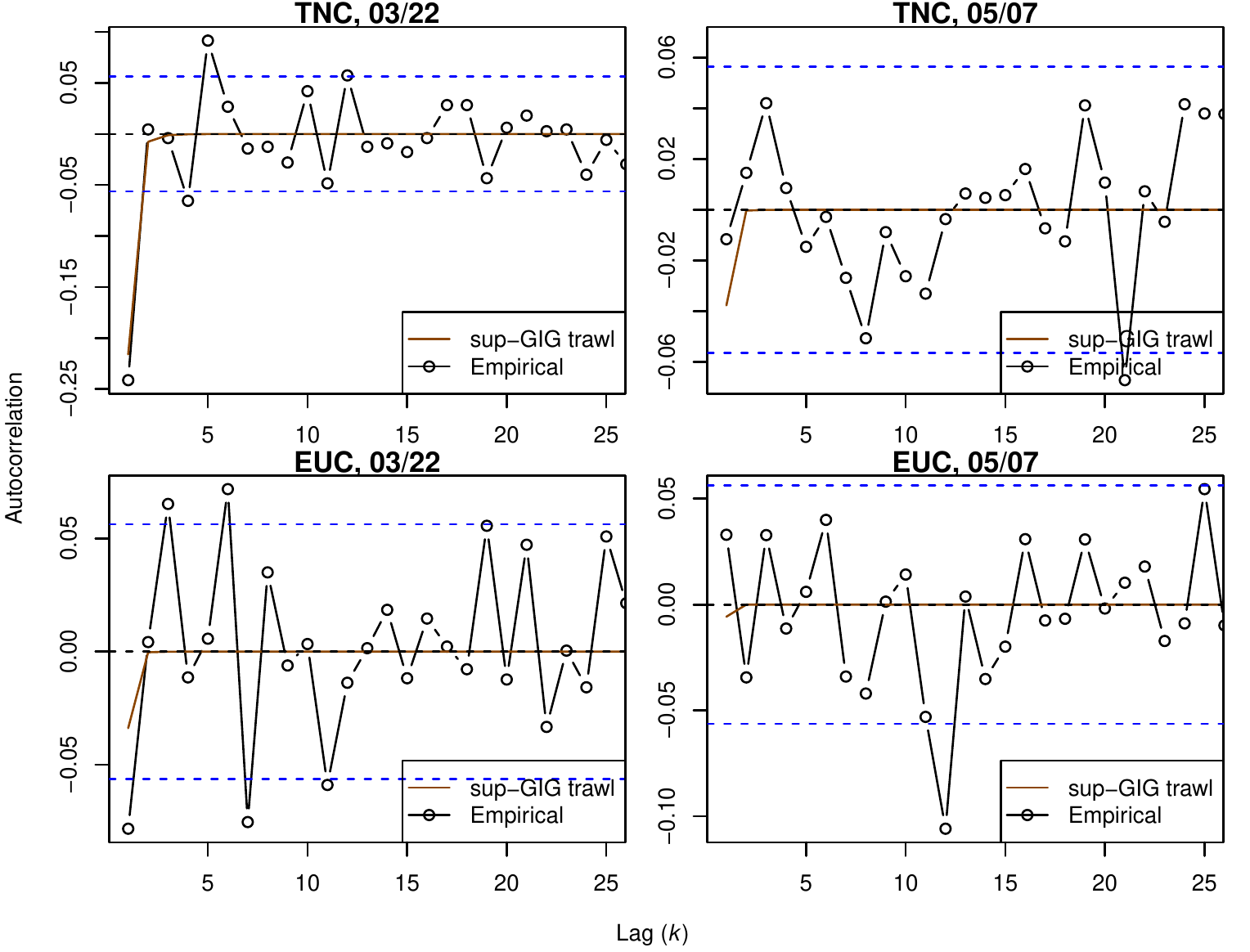}
        \caption{$\delta=60$ seconds.}
    \end{subfigure}%
\begin{subfigure}[t]{0.48\textwidth}
        \centering
        \includegraphics[width=\textwidth]{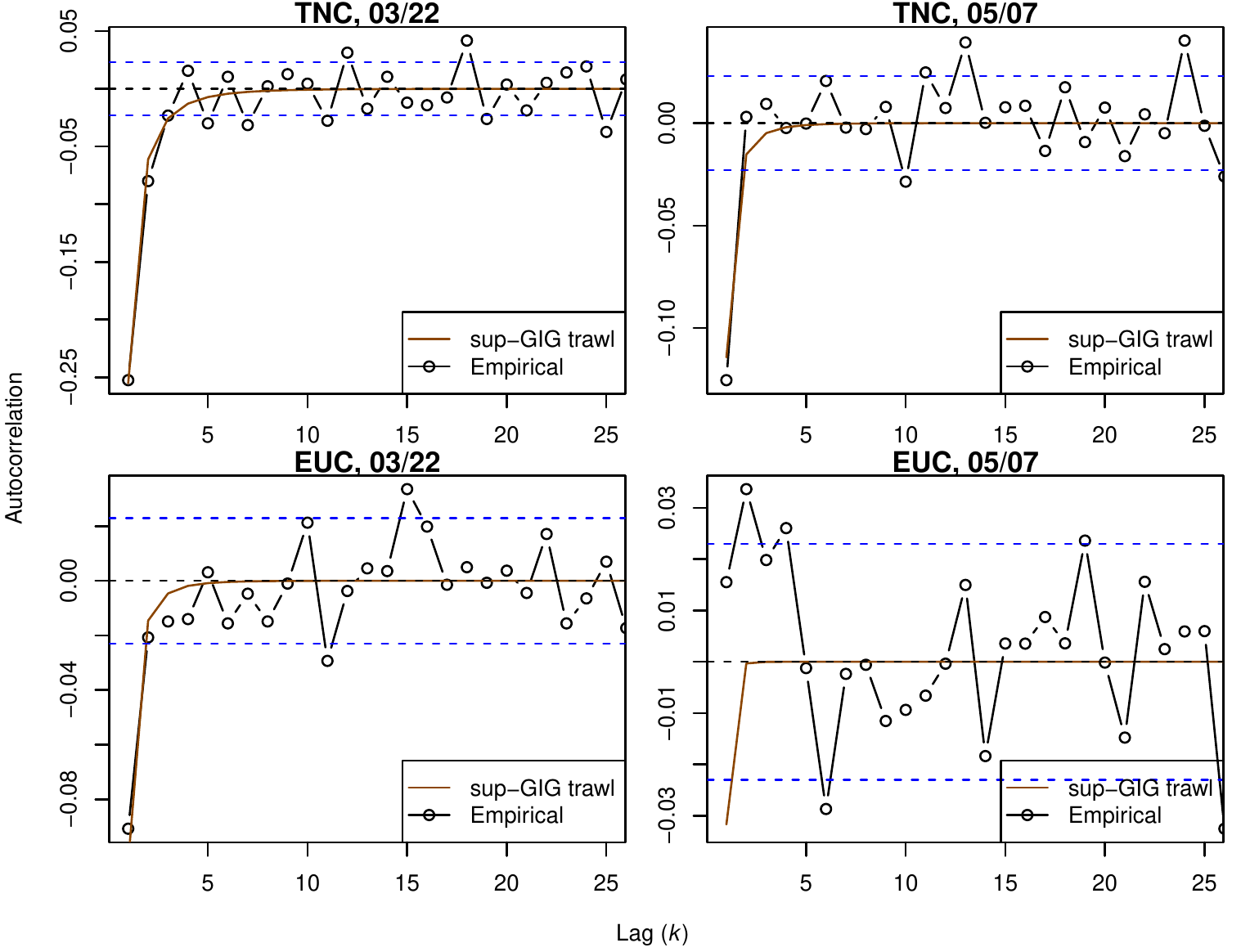}
        \caption{$\delta=10$ seconds.}
    \end{subfigure}\newline
\begin{subfigure}[t]{0.48\textwidth}
        \centering
        \includegraphics[width=\textwidth]{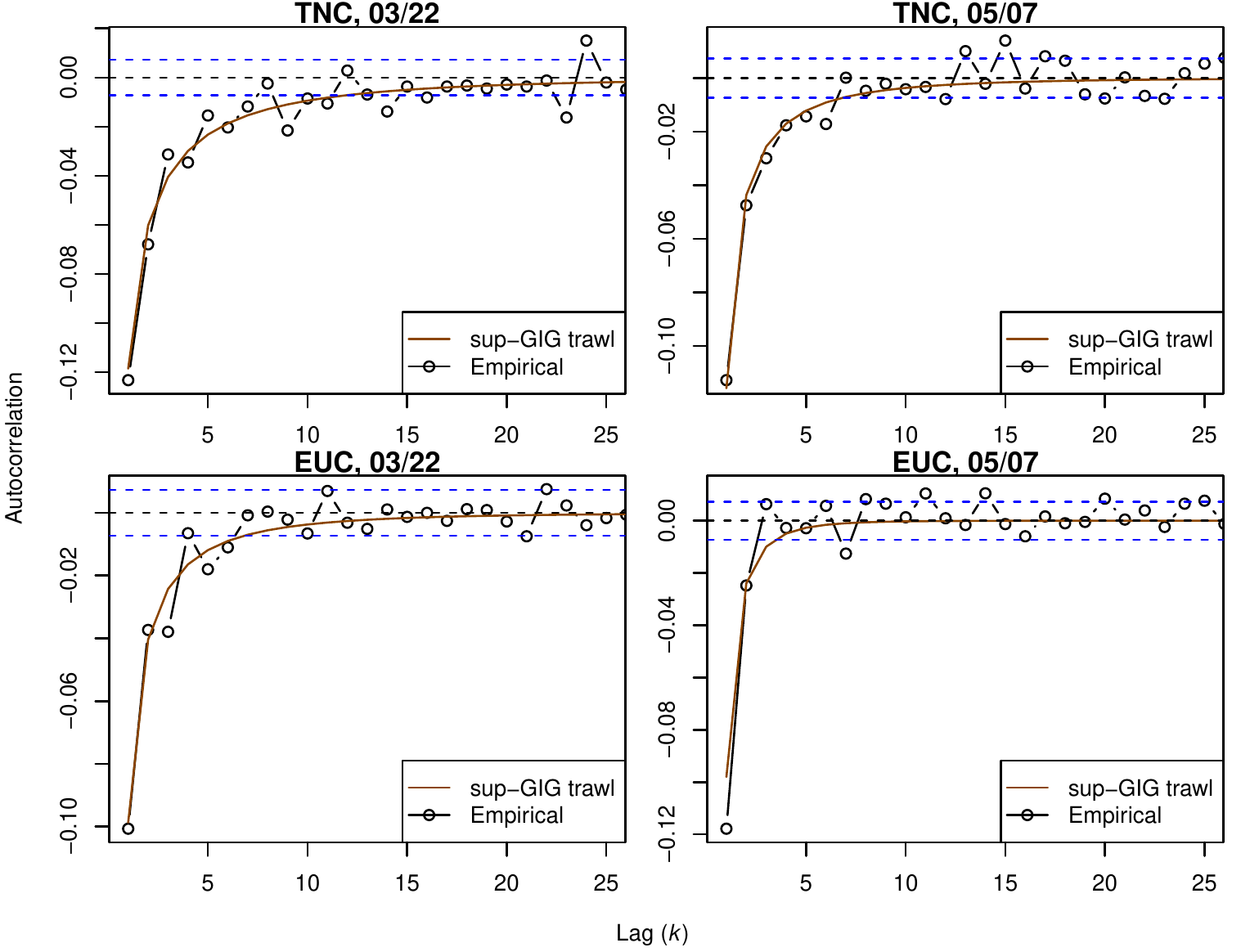}
        \caption{$\delta=1$ second.}
    \end{subfigure}%
\begin{subfigure}[t]{0.48\textwidth}
        \centering
        \includegraphics[width=\textwidth]{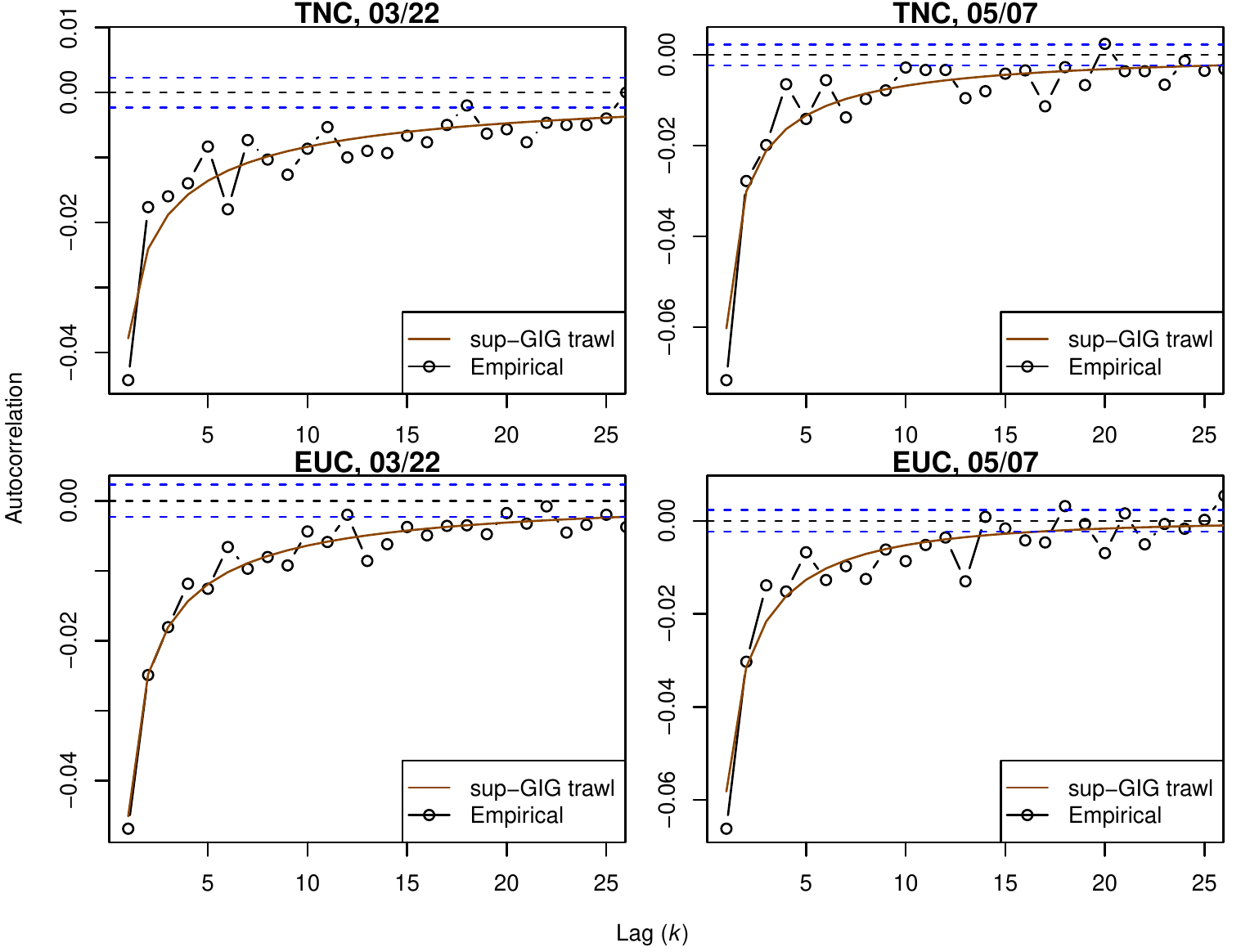}
        \caption{$\delta=0.1$ seconds.}
    \end{subfigure}
\caption{The correlograms for the returns of the four data sets over several
sampling intervals along with the theoretical curves from sup-GIG trawl. The
dashed lines are located at $\pm 2/\protect\sqrt{T/\protect\delta }$.
Code: \texttt{Moment\_Inference\_v2.0.R}.}
\label{fig.:Empirical ACF with fitting}
\end{figure}
For a larger $\delta $, the empirical returns look almost uncorrelated
(insignificant from being $0$) except for TNC on March 22, but the sup-GIG
trawl still captures this anomaly at the first lag. As $\delta $ becomes
smaller, those negative correlations become more significant; even though
the exponential trawl and the sup-$\Gamma $ trawl (not shown in Figure \ref%
{fig.:Empirical ACF with fitting}) can depict the shape of the
autocorrelation, only sup-GIG trawl can fit the first few lags.

As a summary, the sup-GIG trawl (or essentially the sup-$\Gamma ^{-1}$
trawl) performs better than the other two trawls in every aspects. These
empirical analyses demonstrate the descriptive power of our proposed model
for the futures data.

\begin{remark}
We now criticize the insufficient part of our proposed model. A plot (not
shown) of the counting process of price moves for our four data sets will
clearly show a non-linear increasing pattern that disobeys the linearity
described by equation (\ref{Expectation of the power variation}). This
non-linear pattern can be attributed to the well-known diurnal time-varying
levels of trading activity. For the same contract, its two counting process
plots look alike (after rescaling) across different trading days.

We are currently exploring methods that can adjust the model to deal with
these effects, hoping to report on them shortly. It will involve the use of
two independent stochastic time changes for the positive events and the
negative events. A special case on the Skellam L\'{e}vy process using this
ideas has been addressed in \cite{KerssLeonenkoSikorskii(14)}.
\end{remark}

\section{Conclusion\label{sect:conclusion}}

We propose a novel and simple model that can adequately capture some of the
important features of high frequency financial data. It is able to deal with
the dependence in price changes measured over three different orders of
magnitude of time intervals. The model is directly formulated in terms of
the price impact curve (or trawl function). It has a c\`{a}dl\`{a}g price
process that is a piecewise constant semimartingale with finite activity,
finite variation and no Brownian motion component.

However, we need to emphasize that, the proposed model in this paper is just
an initial step. Even though we emphasize the discreteness and the
fleetingness in the movements of the price process, we have been assuming a
simple structure so far with no time-varying features. We will shortly
report on how to generalize this model to the more realistic case using a
stochastic time-change.

Our model provides a good description to the empirical data, while we
majorly focus on the trade prices, which is not always immediately tradable.
For market practitioners who sit either on the buy side or the sell side,
they might consider to apply the proposed model on either the ask price or
bid price, so our model is much more widely applicable than the cases we
report here.

\section{Acknowledgements}

We thank Mikkel Bennedsen, Per Mykland, Mikkel Plagborg-M\o ller, Almut
Veraart, Shihao Yang and various seminar groups for their comments on an
earlier draft.

{\small
\bibliographystyle{chicago}
\bibliography{../Main-AddOn,../neil}
}

\appendix%

\section{\label{Sec.: Proof}Proofs and details}

\begin{proof}[Details of Remark \protect\ref{Cor.: Being a semimartingale}]
To see this, we first argue that $P_{t}$ is a semimartingale with respect to
the \emph{complete data filtration} $\mathcal{F}_{t}\vee \mathcal{S}_{t}$,
which includes the history of the price process itself ($\mathcal{F}_{t}$)
and the history of all the hidden activities of events ($\mathcal{S}_{t}$).
Precisely, $\mathcal{S}_{t}$ is the natural filtration generated by the
process of \emph{random set} $S_{t}$, which consists of all the surviving
events $\left( q,y\right) $ in the trawl at time $t$, where $q\leq t$ is the
original arrival time of the event and $y$ is its size. Then clearly $%
L\left( A_{t}\right) =\sum_{\left( q,y\right) \in S_{t}}y$ is a c\`{a}dl\`{a}%
g adapted process (w.r.t. $\mathcal{S}_{t}$) of locally bounded variation if
the underlying L\'{e}vy basis has finite activities.

Denote the natural filtration generated by the path of $L\left( B_{t}\right)
$ as $\mathcal{L}_{t}$. Then from the definition of $P_{t}$ the complete
data information $\mathcal{F}_{t}\vee \mathcal{S}_{t}$ must be the same as $%
\mathcal{L}_{t}\vee \mathcal{S}_{t}\vee \sigma \left( V_{0}\right) $---the
path of $L\left( B_{t}\right) $ will be completely revealed under $\mathcal{F%
}_{t}\vee \mathcal{S}_{t}$, where $\left\{ \mathcal{L}_{t}\right\} $, $%
\left\{ \mathcal{S}_{t}\right\} $ and $V_{0}$ are completely independent to
each other. Thus,%
\begin{equation*}
M_{t}\triangleq L\left( B_{t}\right) -b\left( \sum_{y\in
\mathbb{Z}
\backslash \left\{ 0\right\} }y\nu \left( y\right) \right) t\in \mathcal{L}%
_{t}\subseteq \mathcal{F}_{t}\vee \mathcal{S}_{t}
\end{equation*}%
must be a martingale w.r.t. $\mathcal{F}_{t}\vee \mathcal{S}_{t}$ because%
\begin{equation*}
\mathbb{E}\left( M_{t}|\mathcal{F}_{s},\mathcal{S}_{s}\right) =\mathbb{E}%
\left( M_{t}|\mathcal{L}_{s},\mathcal{S}_{s},V_{0}\right) =\mathbb{E}\left(
M_{t}|\mathcal{L}_{s}\right) =M_{s},
\end{equation*}%
where the second equality follows from the independence between $\left\{
\mathcal{L}_{t}\right\} $, $\left\{ \mathcal{S}_{t}\right\} $ and $V_{0}$.

Write%
\begin{equation*}
P_{t}=M_{t}+Q_{t},\ \ \ \ Q_{t}\triangleq V_{0}+L\left( A_{t}\right)
+b\left( \sum_{y\in
\mathbb{Z}
\backslash \left\{ 0\right\} }y\nu \left( y\right) \right) t.
\end{equation*}%
As $V_{0}$ can be revealed under $\mathcal{F}_{0}$ and $\mathcal{S}_{0}$, it
is trivially in $\mathcal{F}_{t}\vee \mathcal{S}_{t}$, too. Then $Q_{t}$ is
also a c\`{a}dl\`{a}g adapted process (w.r.t. $\mathcal{F}_{t}\vee \mathcal{S%
}_{t}$) of locally bounded variation. We then conclude that $P_{t}$ is a
semimartingale w.r.t. $\mathcal{F}_{t}\vee \mathcal{S}_{t}$. As the property
of being a semimartingale is preserved under shrinking the filtration, $%
P_{t} $ is a semimartingale w.r.t. $\mathcal{F}_{t}\subseteq \mathcal{F}%
_{t}\vee \mathcal{S}_{t}$.
\end{proof}

\begin{proof}[Proof of Theorem \protect\ref{thm:dist of returns}]
We partition $C_{t}$ and $C_{0}$ into three disjoint sets, one of which is
in common:%
\begin{equation*}
C_{t}=\left( C_{t}\cap C_{0}\right) \cup \left( C_{t}\backslash C_{0}\right)
,\ \ \ \ C_{0}=\left( C_{t}\cap C_{0}\right) \cup \left( C_{0}\backslash
C_{t}\right) ,
\end{equation*}%
so this means that%
\begin{equation*}
P_{t}-P_{0}=L\left( C_{t}\backslash C_{0}\right) -L(C_{0}\backslash C_{t}).
\end{equation*}%
$L\left( C_{t}\backslash C_{0}\right) $ is clearly independent of $%
L(C_{0}\backslash C_{t})$ due to the independence property of the L\'{e}vy
basis and the disjointedness between $C_{t}\backslash C_{0}$ and $%
C_{0}\backslash C_{t}$.

For any $t\geq 0$,
\begin{eqnarray*}
C_{t}\backslash C_{0} &=&\left( A_{t}\backslash A\right) \cup B_{t}=\left(
A_{t}\backslash A\right) \cup \left( \lbrack 0,b)\times (0,t]\right) \\
C_{0}\backslash C_{t} &=&A\backslash A_{t}, \\
leb\left( C_{t}\backslash C_{0}\right) &=&leb\left( A_{t}\backslash A\right)
+bt, \\
leb\left( C_{0}\backslash C_{t}\right) &=&leb\left( A\backslash A_{t}\right)
=leb\left( A_{t}\backslash A\right) .
\end{eqnarray*}%
Then%
\begin{eqnarray*}
C\left( \theta \ddagger P_{t}-P_{0}\right) &=&C\left( \theta \ddagger
L\left( C_{t}\backslash C_{0}\right) \right) +C\left( -\theta \ddagger
L\left( C_{0}\backslash C_{t}\right) \right) , \\
&=&leb(C_{t}\backslash C_{0})C\left( \theta \ddagger L_{1}\right)
+leb(C_{0}\backslash C_{t})C\left( -\theta \ddagger L_{1}\right) \\
&=&btC\left( \theta \ddagger L_{1}\right) +leb(A_{t}\backslash A)\left(
C\left( \theta \ddagger L_{1}\right) +C\left( -\theta \ddagger L_{1}\right)
\right) .
\end{eqnarray*}%
For any random variable $X$ we always have%
\begin{equation*}
\kappa _{j}\left( X\right) =\dfrac{1}{\mathrm{i}^{j}}\left. \dfrac{\partial
^{j}}{\partial ^{j}\theta }C\left( \theta \ddagger X\right) \right\vert
_{\theta =0},
\end{equation*}%
so using the equation above it is clear that%
\begin{equation*}
\kappa _{j}(P_{t}-P_{0})=\left( bt+leb(A_{t}\backslash A)\left(
1+(-1)^{j}\right) \right) \kappa _{j}(L_{1}),
\end{equation*}%
which is the required result.
\end{proof}

\begin{proof}[Proof of Theorem \protect\ref{Thm.: Jumping distribution}]
For each $y\in
\mathbb{Z}
\backslash \left\{ 0\right\} $, the price process has a jump with size $y$
if and only if either one event with size $y$ arrives or one event with size
$-y$ departures---thanks to the monotonicity of $d$. Thus, the probability
of the arrival event can be characterized by the non-zero probability of a
Poisson random variable with intensity
\begin{equation*}
\nu \left( y\right) leb\left( D_{t}\backslash D_{t-\mathrm{d}t}\right)
\approx \nu \left( y\right) \mathrm{d}t;
\end{equation*}%
on the other hand, the probability of the departure event can be similarly
depicted by the non-zero probability of a Poisson random variable with
intensity%
\begin{equation*}
\nu \left( -y\right) leb\left( A_{t-\mathrm{d}t}\backslash A_{t}\right)
\approx \nu \left( -y\right) \left( 1-b\right) \mathrm{d}t.
\end{equation*}%
Therefore, by noting that $\mathbb{P}\left( X>0\right) =1-e^{-\lambda
}\approx \lambda $ for $X\backsim \mathrm{Pois}\left( \lambda \right) $ and
small $\lambda $, we have%
\begin{eqnarray*}
\mathbb{P}\left( \Delta P_{t}=y|\Delta P_{t}\neq 0\right) &=&\dfrac{\mathbb{P%
}\left( \Delta P_{t}=y\right) }{\sum_{y\in
\mathbb{Z}
\backslash \left\{ 0\right\} }\mathbb{P}\left( \Delta P_{t}=y\right) } \\
&=&\dfrac{\nu \left( y\right) \mathrm{d}t+\nu \left( -y\right) \left(
1-b\right) \mathrm{d}t}{\sum_{y\in
\mathbb{Z}
\backslash \left\{ 0\right\} }\left( \nu \left( y\right) \mathrm{d}t+\nu
\left( -y\right) \left( 1-b\right) \mathrm{d}t\right) } \\
&=&\dfrac{\nu \left( y\right) +\nu \left( -y\right) \left( 1-b\right) }{%
\left( 2-b\right) \left\Vert \nu \right\Vert }.
\end{eqnarray*}
\end{proof}

\begin{proof}[Proof of Theorem \protect\ref{Thm.: dependence of returns}]
We will use the following straightforward result on the increments of a
process to prove Theorem \ref{Thm.: dependence of returns}.

\begin{lemma}
\label{Remark about covariances} Suppose that $Z_{t},$ for $t\in
\mathbb{R}
$, has covariance stationary increments. Then for $\delta >0$ and $%
k=1,2,3,...$%
\begin{eqnarray*}
\gamma _{k} &\triangleq &\func{Cov}\left( Z_{\left( k+1\right) \delta
}-Z_{k\delta },Z_{\delta }-Z_{0}\right) \\
&=&\dfrac{1}{2}\func{Var}\left( Z_{\left( k+1\right) \delta }-Z_{k\delta
}\right) -\func{Var}\left( Z_{k\delta }-Z_{0}\right) +\dfrac{1}{2}\func{Var}%
\left( Z_{\left( k-1\right) \delta }-Z_{0}\right) .
\end{eqnarray*}
\end{lemma}

\begin{proof}
First note that%
\begin{eqnarray*}
\mathrm{Var}\left( Z_{(k+1)\delta }-Z_{0}\right) &=&\mathrm{Var}\left(
\left( Z_{(k+1)\delta }-Z_{k\delta }\right) +(Z_{k\delta }-Z_{0})\right) \\
&=&\mathrm{Var}\left( Z_{\delta }-Z_{0}\right) +\mathrm{Var}\left(
Z_{k\delta }-Z_{0}\right) +2\mathrm{Cov}\left( Z_{(k+1)\delta }-Z_{k\delta
},Z_{k\delta }-Z_{0}\right) .
\end{eqnarray*}%
By rearranging, we have%
\begin{equation*}
2\gamma _{k}^{\ast }\triangleq 2\mathrm{Cov}\left( Z_{(k+1)\delta
}-Z_{k\delta },Z_{k\delta }-Z_{0}\right) =\mathrm{Var}\left( Z_{(k+1)\delta
}-Z_{0}\right) -\mathrm{Var}\left( Z_{\delta }-Z_{0}\right) -\mathrm{Var}%
\left( Z_{k\delta }-Z_{0}\right) .
\end{equation*}%
If $k\geq 2$, then
\begin{eqnarray*}
2\gamma _{k}^{\ast } &=&2\mathrm{Cov}\left( Z_{(k+1)\delta }-Z_{k\delta
},Z_{k\delta }-Z_{0}\right) \\
&=&2\mathrm{Cov}\left( Z_{(k+1)\delta }-Z_{k\delta },Z_{k\delta }-Z_{\delta
}\right) +2\mathrm{Cov}\left( Z_{(k+1)\delta }-Z_{k\delta },Z_{\delta
}-Z_{0}\right) =2\gamma _{k-1}^{\ast }+2\gamma _{k}.
\end{eqnarray*}%
Hence,%
\begin{eqnarray*}
\gamma _{k} &=&\dfrac{2\gamma _{k}^{\ast }-2\gamma _{k-1}^{\ast }}{2}=\dfrac{%
1}{2}\left(
\begin{array}{c}
\mathrm{Var}\left( Z_{(k+1)\delta }-Z_{0}\right) -\mathrm{Var}\left(
Z_{\delta }-Z_{0}\right) -\mathrm{Var}\left( Z_{k\delta }-Z_{0}\right) \\
-\left( \mathrm{Var}\left( Z_{k\delta }-Z_{0}\right) -\mathrm{Var}\left(
Z_{\delta }-Z_{0}\right) -\mathrm{Var}\left( Z_{(k-1)\delta }-Z_{0}\right)
\right)%
\end{array}%
\right) \\
&=&\dfrac{1}{2}\mathrm{Var}\left( Z_{(k+1)\delta }-Z_{0}\right) -\mathrm{Var}%
\left( Z_{k\delta }-Z_{0}\right) +\dfrac{1}{2}\mathrm{Var}\left(
Z_{(k-1)\delta }-Z_{0}\right) ,
\end{eqnarray*}%
which is the required result.
\end{proof}

Combining Lemma \ref{Remark about covariances} and Theorem \ref{thm:dist of
returns} gives us%
\begin{eqnarray*}
\gamma _{k} &=&\frac{1}{2}\left( \mathrm{Var}\left( P_{(k+1)\delta
}-P_{0}\right) -2\mathrm{Var}\left( P_{k\delta }-P_{0}\right) +\mathrm{Var}%
\left( P_{(k-1)\delta }-P_{0}\right) \right) \\
&=&\frac{1}{2}\left( b\left( k+1\right) \delta +2leb\left( A_{\left(
k+1\right) \delta }\backslash A\right) -2\left( bk\delta +2leb\left(
A_{k\delta }\backslash A\right) \right) +b\left( k-1\right) \delta
+2leb\left( A_{\left( k-1\right) \delta }\backslash A\right) \right) \kappa
_{2}\left( L_{1}\right) \\
&=&\left( leb\left( A_{\left( k+1\right) \delta }\backslash A\right)
-2leb\left( A_{k\delta }\backslash A\right) +leb\left( A_{\left( k-1\right)
\delta }\backslash A\right) \right) \kappa _{2}\left( L_{1}\right) , \\
\rho _{k} &=&\dfrac{\gamma _{k}}{\func{Var}\left( P_{\delta }-P_{0}\right) }=%
\dfrac{leb\left( A_{\left( k+1\right) \delta }\backslash A\right)
-2leb\left( A_{k\delta }\backslash A\right) +leb\left( A_{\left( k-1\right)
\delta }\backslash A\right) }{b\delta +2leb\left( A_{\delta }\backslash
A\right) }.
\end{eqnarray*}
\end{proof}

\begin{proof}[Proof of Corollary \protect\ref{Cor.: Negative Correlation}]
From Proposition \ref{prop: stat trawl} we have%
\begin{equation*}
\dfrac{\partial }{\partial t}leb\left( A_{t}\cap A\right) =-\left( d\left(
-t\right) -b\right) ,
\end{equation*}%
so mean value theorem states that, for any $0\leq t_{1}<t_{2}<t_{3}$, there
exist $t_{23}\in \left( t_{2},t_{3}\right) $ and $t_{12}\in \left(
t_{1},t_{2}\right) $ such that%
\begin{eqnarray}
\dfrac{leb\left( A_{t_{3}}\cap A\right) -leb\left( A_{t_{2}}\cap A\right) }{%
t_{3}-t_{2}} &=&-\left( d\left( -t_{23}\right) -b\right)  \notag \\
&\leq &-\left( d\left( -t_{12}\right) -b\right)
\label{Collorary strict key ineq.} \\
&=&\dfrac{leb\left( A_{t_{2}}\cap A\right) -leb\left( A_{t_{1}}\cap A\right)
}{t_{2}-t_{1}},  \notag
\end{eqnarray}%
where the second inequality follows from the monotonicity of $d$ and $%
t_{12}<t_{23}$. This proves that $leb\left( A_{t}\cap A\right) $ is a convex
function of $t$. Hence, equation (\ref{trawl increament area}) implies%
\begin{eqnarray}
&&leb(A_{(k+1)\delta }\backslash A)-2leb(A_{k\delta }\backslash
A)+leb(A_{(k-1)\delta }\backslash A)  \notag \\
&=&-leb\left( A_{\left( k+1\right) \delta }\cap A\right) +2leb\left(
A_{k\delta }\cap A\right) -leb\left( A_{\left( k-1\right) \delta }\cap
A\right) \leq 0,  \label{rho_k < 0}
\end{eqnarray}%
as required.

When $d$ is a strictly increasing function, the inequality (\ref{Collorary
strict key ineq.}) becomes strict, so $leb\left( A_{t}\cap A\right) $
becomes a strictly convect function of $t$, which further makes inequality (%
\ref{rho_k < 0}) strict, as required.
\end{proof}

\begin{proof}[Proof of Theorem \protect\ref{Thm: Expectation of Power
variation}]
Arrivals are in $D_{t}$ and so aggregated to $\Sigma (D_{t};r)$, while
departures only happen at most once due to the monotonicity of $d$. All the
departures are in $G_{t}$ and so aggregated to $\Sigma (G_{t};r)$. Now%
\begin{eqnarray*}
\mathbb{E}\left( \left\{ P\right\} _{t}^{\left[ r\right] }\right) &=&\mathbb{%
E}\left( \Sigma \left( B_{t};r\right) \right) +\mathbb{E}\left( \Sigma
\left( H_{t};r\right) \right) +\mathbb{E}\left( \Sigma \left( G_{t};r\right)
\right) \\
&=&\left( leb(B_{t})+leb(H_{t})+leb(G_{t})\right) \int_{-\infty }^{\infty
}\left\vert y\right\vert ^{r}\nu (\mathrm{d}y) \\
&=&\left( bt+(1-b)t+(1-b)t\right) \int_{-\infty }^{\infty }\left\vert
y\right\vert ^{r}\nu (\mathrm{d}y) \\
&=&\left( 2-b\right) t\int_{-\infty }^{\infty }\left\vert y\right\vert
^{r}\nu (\mathrm{d}y),
\end{eqnarray*}%
where the third equality follows from%
\begin{equation*}
leb\left( G_{t}\right) =leb\left( H_{t}\cup A\right) -leb\left( A_{t}\right)
=leb\left( H_{t}\right) +leb\left( A\right) -leb\left( A_{t}\right)
=leb\left( H_{t}\right) =\left( 1-b\right) t.
\end{equation*}
\end{proof}

\begin{proof}[Proof of Proposition \protect\ref{Thm: RV expectation}]
Stationarity of returns and the definition that $\delta _{n}=T/n$ imply%
\begin{eqnarray*}
\mathbb{E}\left( RV^{\left( n\right) }\right) &=&\sum_{k=1}^{n}\mathbb{E}%
\left( P_{k\delta _{n}}-P_{\left( k-1\right) \delta _{n}}\right) ^{2}=n\func{%
Var}\left( P_{\delta _{n}}-P_{0}\right) +n\left( \mathbb{E}\left( P_{\delta
_{n}}-P_{0}\right) \right) ^{2} \\
&=&n\left( b\delta _{n}+2leb\left( A_{\delta _{n}}\backslash A\right)
\right) \kappa _{2}\left( L_{1}\right) +n\left( b\delta _{n}\kappa
_{1}\left( L_{1}\right) \right) ^{2} \\
&=&\left( b+2\dfrac{leb\left( A_{\delta _{n}}\backslash A\right) }{\delta
_{n}}\right) T\kappa _{2}\left( L_{1}\right) +b^{2}T\delta _{n}\kappa
_{1}^{2}\left( L_{1}\right).
\end{eqnarray*}
\end{proof}

\begin{proof}[Details of Remark \protect\ref{Rmk: b cannot be estimated
seperately}]
For any $r\geq 0$ plug-in (\ref{Levy measure moment estimate in terms of b})
into the left-hand side of (\ref{Moment equation from power variation}). Then%
\begin{eqnarray*}
\left( 2-b\right) \sum_{y\in
\mathbb{Z}
\backslash \left\{ 0\right\} }\left\vert y\right\vert ^{r}\widehat{\nu
\left( y\right) } &=&\left( 2-b\right) \sum_{y\in
\mathbb{Z}
\backslash \left\{ 0\right\} }\left\vert y\right\vert ^{r}\left( \dfrac{\hat{%
\alpha}_{y}-\left( 1-b\right) \hat{\alpha}_{-y}}{\left( 2-b\right) b}\hat{%
\beta}_{0}\right) \\
&=&\dfrac{\sum_{y\in
\mathbb{Z}
\backslash \left\{ 0\right\} }\left\vert y\right\vert ^{r}\hat{\alpha}%
_{y}-\left( 1-b\right) \sum_{y\in
\mathbb{Z}
\backslash \left\{ 0\right\} }\left\vert y\right\vert ^{r}\hat{\alpha}_{-y}}{%
b}\hat{\beta}_{0}=\sum_{y\in
\mathbb{Z}
\backslash \left\{ 0\right\} }\left\vert y\right\vert ^{r}\hat{\alpha}_{y}%
\hat{\beta}_{0},
\end{eqnarray*}%
which has nothing to do with parameter $b$.
\end{proof}

\section{\label{Sec.: IFFT Details}Computing probability mass functions of
price changes}

Let $a_{1},...,a_{n}$ be non-zero integers. We will demonstrate how the
inverse fast Fourier transform (IFFT) can be used to calculate $%
p_{y}\triangleq \mathbb{P}\left( Y=y\right) $ of $Y\triangleq
\sum_{k=1}^{n}a_{k}X_{k}\in
\mathbb{Z}
$, where $X_{k}$'s are independent Poisson random variables with intensities
$\lambda _{k}$.

The characteristic function of $Y$ is:%
\begin{equation*}
\varphi \left( \theta \ddagger Y\right) \triangleq \mathbb{E}\left( e^{%
\mathrm{i}\theta Y}\right) =\mathbb{E}\left( e^{\sum_{k=1}^{n}\mathrm{i}%
\theta a_{k}X_{k}}\right) =\prod_{k=1}^{n}\varphi \left( \theta
a_{k}\ddagger X_{k}\right) =\prod_{k=1}^{n}\exp \left( \lambda _{k}\left( e^{%
\mathrm{i}\theta a_{k}}-1\right) \right) .
\end{equation*}%
As $Y$ is discrete, the discrete IFFT can be used to get $p_{y}$. Note that $%
\varphi \left( \theta \ddagger Y\right) =\sum_{y=-\infty }^{\infty }e^{%
\mathrm{i}\theta y}p_{y}$, so the inverse Fourier transform is justified by,
for $y=0,1,2,...$, as $N\rightarrow \infty $,%
\begin{eqnarray*}
\dfrac{1}{N}\sum_{k=0}^{N-1}\varphi \left( -\dfrac{2\pi k}{N}\ddagger
Y\right) e^{\mathrm{i}2\pi ky/N} &=&\dfrac{1}{N}\sum_{k=0}^{N-1}\sum_{y^{%
\prime }=-\infty }^{\infty }p_{y^{\prime }}e^{-\mathrm{i}2\pi ky^{\prime
}/N}e^{\mathrm{i}2\pi ky/N} \\
&=&\dfrac{1}{N}\sum_{y^{\prime }=-\infty }^{\infty }p_{y^{\prime
}}\sum_{k=0}^{N-1}e^{\mathrm{i}2\pi k\left( y-y^{\prime }\right)
/N}\rightarrow \sum_{y^{\prime }=-\infty }^{\infty }p_{y^{\prime
}}1_{\left\{ y=y^{\prime }\right\} }=p_{y}, \\
\dfrac{1}{N}\sum_{k=0}^{N-1}\varphi \left( \dfrac{2\pi k}{N}\ddagger
Y\right) e^{\mathrm{i}2\pi ky/N} &=&\dfrac{1}{N}\sum_{y^{\prime }=-\infty
}^{\infty }p_{y^{\prime }}\sum_{k=0}^{N-1}e^{\mathrm{i}2\pi k\left(
y+y^{\prime }\right) /N}\rightarrow \sum_{y^{\prime }=-\infty }^{\infty
}p_{y^{\prime }}1_{\left\{ y=-y^{\prime }\right\} }=p_{-y},
\end{eqnarray*}%
where the approximation here comes from the Riemann sum%
\begin{equation*}
\dfrac{1}{N}\sum_{k=0}^{N-1}e^{\mathrm{i}2\pi k\theta ^{\prime
}/N}=\int_{0}^{1}e^{\mathrm{i}2\pi \theta \theta ^{\prime }}\mathrm{d}\theta
+O\left( N^{-1}\right) =1_{\left\{ \theta ^{\prime }=0\right\} }+O\left(
N^{-1}\right) .
\end{equation*}

Hence, the IFFT will take the input of%
\begin{equation*}
\left( \varphi \left( 0\ddagger Y\right) ,\varphi \left( -2\pi /N\ddagger
Y\right) ,...,\varphi \left( -\dfrac{2\pi \left( N-1\right) }{N}\ddagger
Y\right) \right) ^{T}
\end{equation*}%
and give the output as $\left( p_{0},...,p_{N-1}\right) ^{T}$ approximately;
similarly, with the input of%
\begin{equation*}
\left( \varphi \left( 0\ddagger Y\right) ,\varphi \left( 2\pi /N\ddagger
Y\right) ,...,\varphi \left( \dfrac{2\pi \left( N-1\right) }{N}\ddagger
Y\right) \right) ,
\end{equation*}%
the IFFT will give the output as $\left( p_{0},p_{-1},...,p_{-\left(
N-1\right) }\right) ^{T}$ approximately.

In Figure \ref{fig.:Empirical log-histograms of jumps}, we take $N=60$ in
order to accurately compute $p_{y}$ for $y\in \left\{ -30,...,30\right\} $.

\section{\label{Sec.: Data cleaning}Cleaning of the empirical data}

Here we discuss the preprocessing procedures for the raw empirical data. For
each data set, our database has the current bid price (\texttt{bid}), bid
size (\texttt{bidsz}), ask price (\texttt{ask}), ask size (\texttt{asksz}),
trade price (\texttt{trade}), trade volume (\texttt{tradesz}) and the record
logging time on the data server (\texttt{log\_t}). The following events will
be logged into the raw data set chronologically:

\begin{enumerate}
\item A change of \texttt{bid} and \texttt{bidsz} (or \texttt{ask} and
\texttt{asksz}), which will leave missing \texttt{ask}, \texttt{asksz} (or
\texttt{bid}, \texttt{bidsz}), \texttt{trade} and \texttt{tradesz}.

\item A new instance of \texttt{trade} and \texttt{tradesz}, which will
leave missing \texttt{bid}, \texttt{bidsz}, \texttt{ask} and \texttt{asksz}.
This is usually followed by a record that shows the newest \texttt{bid} and
\texttt{ask} status after the trading. Sometimes this updating record will
be combined with its previous trading record.
\end{enumerate}

\paragraph{Step 1: Remove the wrong records (Optional).}

We forward fill the missing values in columns \texttt{bid} and \texttt{ask};
after this, we examine whether the recorded trade price lies in the range
from \texttt{bid} minus a factor \texttt{M} of tick sizes to \texttt{ask}
plus \texttt{M} tick sizes. \texttt{M} here is manually chosen as $9.5$ for
the two EUC data sets, which is a conservative setting and will only remove
those visually inspectable errors. We do not use this step for the two TNC
data sets.

\paragraph{Step 2: Preserve only the trading activities.}

Since in this paper we are only concerned with the dynamics of the trade
prices, we throw out all the other data records that are not directly
associated with a trade, that is, those rows with missing \texttt{trade} and
\texttt{tradesz}.

\paragraph{Step 3-1: Associate a unique price to a time tag.}

Occasionally several data feeds will be pushed into the data server almost
at the same time but perhaps with different prices. Then we iteratively
define a unique price for this particular time tag by the price that is
closest to the price of the previous time tag. Figure \ref{fig.: Multiple
Trading} illustrates this.
\begin{figure}[t]
\centering\includegraphics{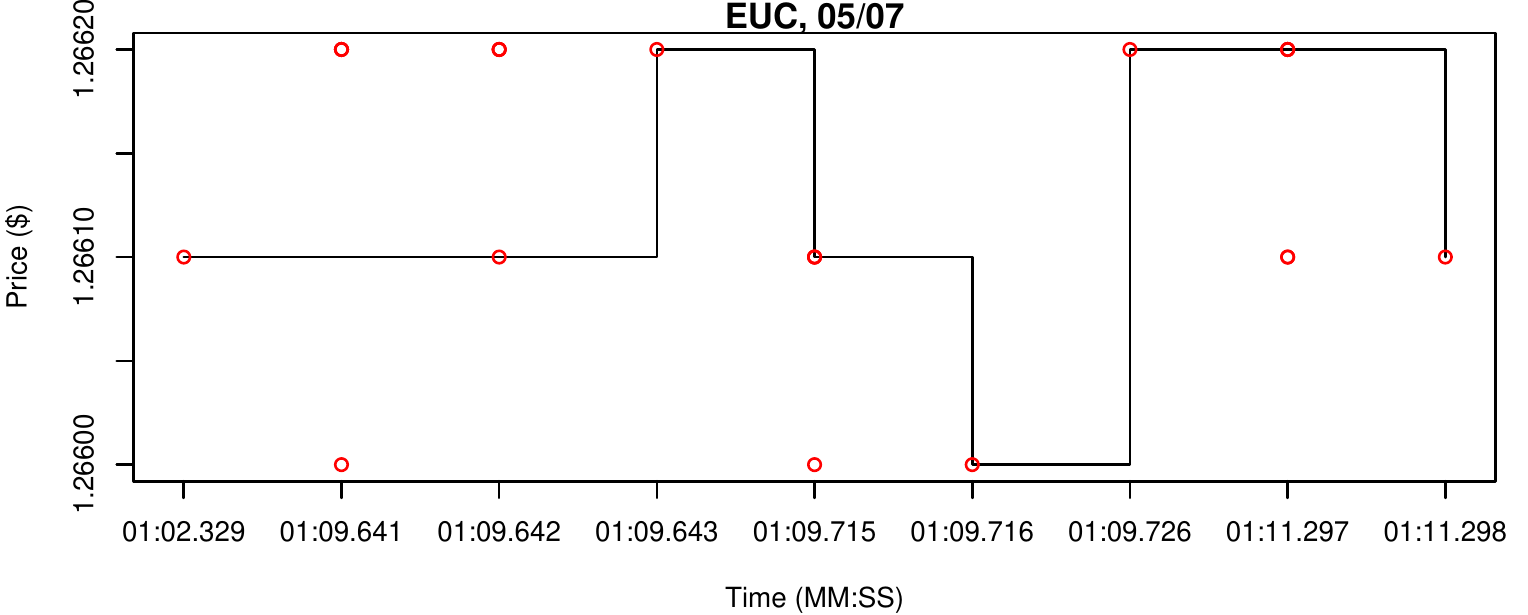}
\caption{An illustration of the definition of a unique price when multiple
data records share the same time tag. The white points illustrate all the
trade prices appeared in the data with the same time tag; the black solid
line represent the unique price we define. This data is EUC1006 between
00:01:02 and 00:01:12 on May 7. Code: \texttt{Price\_Plots.R}.}
\label{fig.: Multiple Trading}
\end{figure}

\paragraph{Step 3-2: Do nothing for an ambiguous case.}

If it happens that there are exactly two trade prices with the same time tag
that are just one tick above and one tick below the previous price---which
we call an ambiguous case (e.g. at time 01:09.641 in Figure \ref{fig.:
Multiple Trading}), then we will use the previous price as the price for the
current time tag.

\paragraph{Step 4: Keep only jumps.}

For our analysis it is sufficient to keep only the columns \texttt{Time} and
\texttt{Price}, such that \texttt{Time} is always increasing without
duplicates while \texttt{Price} have no two adjacent elements that take the
same value. \texttt{Price} is always the value the price process takes
immediately after a jump.

\section{\label{Sec.: Nonparametric d}Non-parametric inference of the trawl
function}

Let $\tilde{d}\left( s\right) $ be the \emph{non-squashed trawl function}
with $\tilde{d}\left( -\infty \right) =0$ such that $d\left( s\right)
\triangleq b+\left( 1-b\right) \tilde{d}\left( s\right) $. Then equation (%
\ref{trawl increament area}) implies $\partial _{\delta }leb\left( A_{\delta
}\backslash A\right) =\left( 1-b\right) \tilde{d}\left( -\delta \right) $.
Hence,%
\begin{equation*}
\dfrac{\partial \widehat{\sigma _{\delta }^{2}}}{\partial \delta }=\left(
\dfrac{b+2\left( 1-b\right) \tilde{d}\left( -\delta \right) }{2-b}\right)
\sum_{y\in
\mathbb{Z}
\backslash \left\{ 0\right\} }y^{2}\hat{\alpha}_{y}\hat{\beta}_{0},\quad
\dfrac{\partial \widehat{\sigma _{\delta }^{2}}}{\partial \delta }\left(
\infty \right) =\left( \dfrac{b}{2-b}\right) \sum_{y\in
\mathbb{Z}
\backslash \left\{ 0\right\} }y^{2}\hat{\alpha}_{y}\hat{\beta}_{0},
\end{equation*}%
which then gives us%
\begin{equation*}
b=\dfrac{2\left( \partial _{\delta }\widehat{\sigma _{\delta }^{2}}\right)
\left( \infty \right) }{\left( \partial _{\delta }\widehat{\sigma _{\delta
}^{2}}\right) \left( 0\right) +\left( \partial _{\delta }\widehat{\sigma
_{\delta }^{2}}\right) \left( \infty \right) },\ \tilde{d}\left( -\delta
\right) =\dfrac{\partial _{\delta }\widehat{\sigma _{\delta }^{2}}-\left(
\partial _{\delta }\widehat{\sigma _{\delta }^{2}}\right) \left( \infty
\right) }{\left( \partial _{\delta }\widehat{\sigma _{\delta }^{2}}\right)
\left( 0\right) -\left( \partial _{\delta }\widehat{\sigma _{\delta }^{2}}%
\right) \left( \infty \right) },
\end{equation*}%
where $\left( \partial _{\delta }\widehat{\sigma _{\delta }^{2}}\right)
\left( 0\right) =\sum_{y\in
\mathbb{Z}
\backslash \left\{ 0\right\} }y^{2}\hat{\alpha}_{y}\hat{\beta}_{0}$.
Therefore, by estimating $\partial _{\delta }\widehat{\sigma _{\delta }^{2}}$
for every $\delta $ the trawl function is revealed non-parametrically.

In practice, it might be demanding to get $\left( \partial _{\delta }%
\widehat{\sigma _{\delta }^{2}}\right) \left( \infty \right) $, the
asymptotic slope of the sample variogram $\widehat{\sigma _{\delta }^{2}}$
against $\delta $, because as $\delta $ being larger, the sample size we use
to calculate $\widehat{\sigma _{\delta }^{2}}$ is getting smaller. Is it
possible to use other moment equations in Theorem \ref{thm:dist of returns}
to identify $b$ rather than through the boundary behavior of $\partial
_{\delta }\widehat{\sigma _{\delta }^{2}}$ for $\delta \rightarrow \infty $?
Unfortunately, the answer is no. $b$ and $d$ are not identifiable if we
neither parameterize $d$ nor adopt a boundary estimation for $b$ at $\delta
\rightarrow \infty $.

To justify this point, assume that one wants to employ all the other
additional moment equations in Theorem \ref{thm:dist of returns} to identify
$b$:%
\begin{eqnarray*}
\kappa _{j}\left( P_{\delta }-P_{0}\right) &=&\left( b\delta +\left(
1+\left( -1\right) ^{j}\right) leb\left( A_{\delta }\backslash A\right)
\right) \kappa _{j}\left( L_{1}\right) \\
&=&\left( b\delta +\left( 1+\left( -1\right) ^{j}\right) leb\left( A_{\delta
}\backslash A\right) \right) \sum_{y\in
\mathbb{Z}
\backslash \left\{ 0\right\} }y^{j}\nu \left( y\right) ,\ \ \ \ j\geq 3.
\end{eqnarray*}%
Denote the sample $j$-th cumulant with sampling interval $\delta $ as $%
\widehat{\kappa _{j,\delta }}$. Then equation (\ref{Levy measure moment
estimate in terms of b}) implies that%
\begin{eqnarray*}
\dfrac{\partial \widehat{\kappa _{j,\delta }}}{\partial \delta } &=&\left(
b+\left( 1+\left( -1\right) ^{j}\right) \dfrac{\partial }{\partial \delta }%
leb\left( A_{\delta }\backslash A\right) \right) \sum_{y\in
\mathbb{Z}
\backslash \left\{ 0\right\} }y^{j}\dfrac{\hat{\alpha}_{y}-\left( 1-b\right)
\hat{\alpha}_{-y}}{\left( 2-b\right) b}\hat{\beta}_{0} \\
&=&\left( b+\left( 1+\left( -1\right) ^{j}\right) \left( 1-b\right) \tilde{d}%
\left( -\delta \right) \right) \dfrac{\sum_{y\in
\mathbb{Z}
\backslash \left\{ 0\right\} }y^{j}\hat{\alpha}_{y}-\left( 1-b\right)
\sum_{y\in
\mathbb{Z}
\backslash \left\{ 0\right\} }y^{j}\hat{\alpha}_{-y}}{\left( 2-b\right) b}%
\hat{\beta}_{0} \\
&=&\left( b+\left( 1+\left( -1\right) ^{j}\right) \left( 1-b\right) \tilde{d}%
\left( -\delta \right) \right) \dfrac{1-\left( -1\right) ^{j}\left(
1-b\right) }{\left( 2-b\right) b}\sum_{y\in
\mathbb{Z}
\backslash \left\{ 0\right\} }y^{j}\hat{\alpha}_{y}\hat{\beta}_{0} \\
&=&\left\{
\begin{array}{cc}
\sum_{y\in
\mathbb{Z}
\backslash \left\{ 0\right\} }y^{j}\hat{\alpha}_{y}\hat{\beta}_{0} & ,\
\text{for }j\text{ odd} \\
\dfrac{\partial \widehat{\sigma _{\delta }^{2}}/\partial \delta }{\sum_{y\in
\mathbb{Z}
\backslash \left\{ 0\right\} }y^{2}\hat{\alpha}_{y}\hat{\beta}_{0}}%
\sum_{y\in
\mathbb{Z}
\backslash \left\{ 0\right\} }y^{j}\hat{\alpha}_{y}\hat{\beta}_{0} & ,\
\text{for }j\text{ even}%
\end{array}%
\right. ,
\end{eqnarray*}%
which is still again independent of $b$.

\end{document}